\documentclass[10pt]{article}

\usepackage{amsthm}
\usepackage{amsmath}
\usepackage{amsfonts}
\usepackage{amssymb}
\usepackage{amscd}
\usepackage{mathrsfs}
\usepackage{bm}
\usepackage{url}
\usepackage{color}
\usepackage{cite}

\newtheorem{thm}{Theorem}[section]
\newtheorem{cor}[thm]{Corollary}
\newtheorem{lem}[thm]{Lemma}
\newtheorem{prop}[thm]{Proposition}
\newtheorem{conj}[thm]{Conjecture}

\theoremstyle{definition}
\newtheorem{example}[thm]{Example}
\newtheorem{remark}[thm]{Remark}

\newcommand{\R}{\mathbb{R}}
\newcommand{\C}{\mathbb{C}}
\newcommand{\Z}{\mathbb{Z}}
\newcommand{\N}{\mathbb{N}}

\newcommand{\PP}{\mathbb{P}}

\newcommand{\RLCT}{{\rm RLCT}}

\newcommand{\mi}{\!-\!}
\newcommand{\nmin}{\#\!\min}
\newcommand{\ord}{{\rm ord}}

\newcommand{\mon}{{\rm mon}}

\title{Ideal-Theoretic Strategies for \\Asymptotic Approximation of \\Marginal Likelihood Integrals}
\date{}
\author{Shaowei Lin}

\begin{document}
\maketitle

\begin{abstract}
The accurate asymptotic evaluation of marginal likelihood integrals is
a fundamental problem in Bayesian statistics. Following the approach
introduced by Watanabe, we translate this into a problem of computational
algebraic geometry, namely, to determine the real log canonical threshold
of a polynomial ideal, and we present effective methods for solving this
problem. Our results are based on resolution of singularities. They apply to parametric models where the Kullback-Leibler distance is upper and lower bounded by scalar multiples of some sum of squared 
real analytic functions. Such models include finite state discrete models.\\
\noindent {\bf Keywords:} computational algebra, asymptotic approximation, marginal likelihood, learning coefficient, real log canonical threshold 

\end{abstract}

\section{Introduction}
\label{sec:Introduction}

The evaluation of marginal likelihood integrals is essential in model selection and has important applications in areas such as machine learning and computational biology. The exact evaluation of such integrals is a difficult problem \cite{DSS, LSX} and classical approximation formulas usually apply only for smooth models. Recent work by Watanabe and his collaborators \cite{AW, W, W0, YW1, YW} extended these formulas to a broad class of models with singularities. His work also uncovered interesting connections with resolution of singularities in algebraic geometry. The \mbox{goal of this} paper is to systematically study the algebraic geometry behind Watanabe's formulas, and to develop symbolic algebra tools which
allow the user to accurately evaluate the asymptotics of integrals in Bayesian statistics. 

Watanabe showed that the key to understanding a singular model is monomializing the Kullback-Leibler function $K(\omega)$ of the model at the true distribution. While general algorithms exist for monomializing any analytic function \cite{BEV,BM}, applying them to non-polynomial functions such as $K(\omega)$ can be computationally expensive. In practice, many singular models are parametrized by polynomials. Therefore, it is natural to ask if this polynomiality can be exploited in the analysis of such models. For simplicity, we explore this question for \mbox{discrete statistical} models. Our point of departure is to describe the asymptotics of the likelihood integral by the real log canonical threshold of an ideal in a polynomial ring.~More generally, our results will be proved for rings of analytic functions, and they apply to all parametric models where the Kullback-Leibler distance is upper and lower bounded by scalar multiples of a sum of squared 
real analytic functions.

Consider a statistical model $\mathscr{M}$ on a finite discrete space $[k] = $ $\{1, 2, \ldots, k \}$ parametrized by a real analytic map
$p: \Omega \rightarrow \Delta_{k-1}$ where $\Omega$ is a compact
subset of $\R^d$ and $\Delta_{k-1}$ is the probability simplex $\{x \in \R^k: x_i \geq 0$, $\sum x_i = 1 \}$. We assume that $\Omega$ is \emph{semianalytic}, i.e. $\Omega =\{x \in \R^d:g_1(x)\geq 0,$ $\ldots, g_l(x) \geq 0 \}$ is defined by real analytic inequalities. Let $q \in \Delta_{k-1}$ be a point in the model with non-zero entries. Suppose a sample of size $N$ is drawn from the true distribution $q$, and let $U = (U_i)$ denote the vector of relative frequencies for this sample. Let $\varphi:\Omega \rightarrow \R$ be \emph{nearly analytic}, i.e. $\varphi$ is a product $\varphi_a \varphi_s$ of functions where $\varphi_a$ is real analytic and $\varphi_s$ is positive and smooth. Consider a Bayesian prior defined by $|\varphi|$. Priors of this form are discussed in Remark~\ref{rem:PriorAbsoluteValue}. We are interested in~the asymptotics, for large sample sizes $N$, of the marginal likelihood integral
\begin{eqnarray}
\label{eq:Marginal}
 Z(N) = \int_{\Omega} \prod_{i=1}^k p_{i}(\omega)^{NU_{i}} |\varphi(\omega)| \,d\omega.
\end{eqnarray}

The first few terms of the asymptotics of the log likelihood integral $\log Z(N)$ was derived by Watanabe. To state his result, we first recall that the Kullback-Leibler distance $K(\omega)$ between $q$ and $p(\omega)$ is
$$
K(\omega) = \sum_{i=1}^{k} q_i \log \frac{q_i}{p_i(\omega)}.
$$
This function satisfies $K(\omega) \geq 0$ with equality if and only if $p(\omega) = q$.
\begin{thm}[Watanabe{\cite[\S6]{W}}]
\label{thm:Watanabe}
Asymptotically as $N \rightarrow \infty$, 
\begin{eqnarray}
\log Z(N) \,\,=\,\, N\sum_{i=1}^{k} U_i \log q_i - \lambda \log N + (\theta-1)\log \log N\,\, + \,\,\eta_N
\label{eq:LogBayesian}
\end{eqnarray}
where the positive rational number $\lambda$ is the smallest pole of the zeta function
\begin{align}
\label{eq:WatanabeZeta}
  \zeta(z) = \int_\Omega K(\omega)^{-z} |\varphi(\omega)| \,d\omega,
  \quad z \in \C, 
\end{align}
$\theta$ is its multiplicity, and $\eta_N$ is a random variable whose expectation $\mathbb{E}[\eta_N]$ converges to a constant.
\end{thm}
Here, $\lambda$ is known as the \emph{learning coefficient} of the model at the distribution $q$. Because formula (\ref{eq:LogBayesian}) generalizes the Bayesian information criterion \cite{GR, W}, the numbers $\lambda$ and $\theta$ are important in model selection. Indeed, the BIC corresponds to the case $(\lambda, \theta) = (\frac{d}{2}, 1)$ for smooth models. In algebraic geometry, $\lambda$ is also known as the \emph{real log canonical threshold} \cite{S} of $K$, a term that is motivated~by the more familiar complex \emph{log canonical threshold} (see Remark \ref{rem:lct}). We denote this algebraic invariant by $(\lambda,\theta)=\RLCT_\Omega(K;\varphi)$.

These thresholds may be defined for ideals in rings of real-valued analytic functions as well. Given an ideal $I = \langle f_1, \ldots, f_r \rangle$ generated by functions $f_i \not\equiv 0$ which are real analytic on a compact subset $\Omega \subset \R^d$, and a smooth amplitude function $\varphi:\R^d \rightarrow \R$, we consider the zeta function
\begin{eqnarray}
\label{eq:IdealZetaFn}
\zeta(z) = \int_{\Omega} \Big(f_1(\omega)^2+\cdots+f_r(\omega)^2\Big)^{-z/2}  |\varphi(\omega)| \,d\omega.
\end{eqnarray}
We show that if $\varphi$ is nearly analytic, then $\zeta(z)$ has an analytic continuation to the whole complex plane. Its poles are positive rational numbers with a smallest element $\lambda$ which we call the \emph{real log canonical threshold} of $I$ with respect to $\varphi$ over $\Omega$. Let $\theta$ be the multiplicity of $\lambda$ as a pole of $\zeta(z)$ and define $\RLCT_\Omega (I; \varphi)$ to be the pair $(\lambda, \theta)$. Order these pairs such that $(\lambda_1, \theta_1) > (\lambda_2, \theta_2)$ if $\lambda_1 > \lambda_2$, or $\lambda_1=\lambda_2$ and $\theta_1 < \theta_2$. We will show that this pair does not depend on the choice of generators $f_1, \ldots, f_r$ for $I$. In the literature, real log canonical thresholds of ideals are not well-investigated \cite{S}. For this reason, we formally prove many of its properties in Section \ref{sec:RLCT}.

With these definitions on hand, we now state our first main theorem. This result expresses the learning coefficient and its multiplicity directly in terms of the functions $p_1, \ldots, p_k$ parametrizing the model. Geometrically, it says that the learning coefficient is the real log canonical threshold of the fiber $p^{-1}(q) \subset \Omega$. The theorem is computationally very useful especially when the~$p_i$ are polynomials or rational functions, and certain special cases have been applied by Sumio Watanabe and his collaborators \cite{YW1, YW}. Our proof in Section \ref{sec:RLCT} was inspired by a discussion with him. Now, recall that $\varphi = \varphi_a \varphi_s$ is nearly analytic.
\begin{thm}
\label{thm:PolyRLCT}
Let $(\lambda, \theta)$ be the learning coefficient and multiplicity of the model $\mathscr{M}$\! at $q > 0$. Let $I$ denote the ideal $\langle \,p(\omega) - q \,\rangle := \langle \,p_1(\omega)-q_1, \ldots, p_k(\omega)-q_k \rangle$,~and let $\mathcal{V}$ be its zero-locus $\{\omega \in \Omega : p(\omega) = q\} = p^{-1}(q)$. Then,
$$
(2\lambda, \theta) = \min_{x \in \mathcal{V}}\, \RLCT_{\Omega_x} (I;\varphi_a)
$$
where each $\Omega_x$ is a sufficiently small neighborhood of $x$ in $\Omega$. 

More generally, let $K(\omega)$ be any real analytic function on $\Omega$ that is bounded for some constants $c_1, c_2 > 0$ and some real analytic $f_i(\omega)$ over $\Omega$, by
$$
c_1 \sum_{i=1}^k f_i(\omega)^2 \leq K(\omega) \leq c_2 \sum_{i=1}^k f_i(\omega)^2.
$$
Then, the real log canonical threshold $(\lambda,\theta)=\RLCT_\Omega(K;\varphi)$ satisfies
$$
(2\lambda, \theta) = \RLCT_{\Omega} (I;\varphi_a)
$$
where $I$ is the ideal $\langle f_1(\omega), \ldots, f_k(\omega)\rangle$. 
\end{thm}

To prove this theorem and other properties of real log canonical thresholds, we recall Hironaka's theorem on the resolution of singularities \cite{Hi} and develop useful lemmas in Section \ref{sec:ResSing}. Our treatment differs from that of Watanabe \cite{W} in the following way: we 
study the \emph{local} behavior of real log canonical thresholds at points $x$ in the parameter space $\Omega$. In particular, we will be interested in the case where $x$ is on the boundary $\partial \Omega$. Example \ref{ex:RLCTdependsonboundary} is an illustration of how the threshold is affected by the inequalities $g_i \geq 0$ which are active at $x$. This issue can be critical in singular model selection because the parameter space of one model is often contained in the boundary of another that is more complex.

After studying the local thresholds, we then show that the real log canonical threshold \emph{globally} over $\Omega$ is the minimum of local thresholds at points $x$ in $\Omega$. Identifying where these minimum thresholds occur is by itself a difficult problem which we discuss in Section \ref{sec:ResSing}. As a consequence of our results, we write down explicit formulas for the coefficients in asymptotic expansions of \mbox{Laplace integrals}. Our formulas extend those of 
Arnol'd--Guse\u\i n-Zade--Varchenko\cite{AGV} because they apply also to parameter spaces with boundary. Using this expansion to improve approximations of likelihood integrals will be the subject of future work.

Our next aim is to develop tools for computing or bounding real log canonical thresholds of ideals. Section \ref{sec:RLCT} summarizes useful fundamental properties of real log canonical thresholds. In Section \ref{sec:Newton}, we derive local thresholds in nondegenerate cases using an important tool from toric geometry involving Newton polyhedra. This method was invented by Varchenko~\cite{V} and applied to statistical models by \mbox{Watanabe and Yamazaki \cite{YW}}. Their formulas were defined for functions, but we develop extensions of these formulas for ideals. We introduce a new notion of nondegeneracy for ideals, known as \emph{sos-nondegeneracy}, and give the following bound for the real log canonical threshold of an ideal with respect to a monomial amplitude function $\omega^\tau := \omega_1^{\tau_1}\cdots \omega_d^{\tau_d}$. These monomial functions occur frequently when we apply a change of variables to resolve the singularities in a model. Newton polyhedra and their $\tau$-distances are defined in Section \ref{sec:Newton}. 

\begin{thm} \label{thm:SosnondegenerateRLCT} Let $I$ be a finitely generated ideal in the ring of functions which are real analytic on $\Omega$, and suppose the origin $0$ lies in the interior of $\Omega$. Then, for every sufficiently small neighborhood  $\Omega_0$ of the origin, 
$$\RLCT_{\Omega_0} (I; \omega^\tau) \leq (1/l_\tau, \theta_\tau)$$
where $l_\tau$ is the $\tau$-distance of the Newton polyhedron $\mathcal{P}(I)$ and $\theta_\tau$ its multiplicity. Equality occurs when $I$ is monomial or, more generally, sos-nondegenerate.
\end{thm}

This theorem has two main consequences. Firstly, it tells us that the real log canonical threshold of an ideal can be computed by finding a change of variables which monomializes the ideal. Secondly, due to Theorems \ref{thm:Watanabe} and \ref{thm:PolyRLCT}, upper bounds on real log canonical thresholds translate to asymptotic lower bounds on the likelihood integral of a statistical model, which in turn give upper bounds on the stochastic complexity of the model.

Currently, there are no programs for computing real log \mbox{canonical thresholds.} There are applications which compute resolutions of singularities, but our statistical problems are too big for them. We hope that our work is a step \mbox{in bridging} the gap. Some of our tools are implemented in a {\sc Singular} library at
$$ \hbox{\url{https://w3id.org/people/shaoweilin/public/rlct.html}.} $$ 
This library computes the Newton polyhedron of an ideal, computes $\tau$-distances, and checks if an ideal is sos-nondegenerate. Instructions and examples on using the library may be found at the above website.

In summary, the learning coefficient of a statistical model is a useful measure of the model complexity and plays an important role in model selection. Because computing this coefficient often requires careful analysis of the Kullback-Leibler function, we propose an ideal-theoretic approach to make this calculation more tractable. This method has several advantages. Firstly, it directly \mbox{exploits poly}-nomiality in the model parametrization. Second, the real log canonical threshold of an ideal is independent of the choice of generators, and this choice provides flexibility to our computations. Thirdly, it is easier to construct Newton polyhedra for polynomial ideals and to check their nondegeneracy (Proposition \ref{thm:DefRLCT}(3)), than for nonpolynomial Kullback-Leibler functions. We demonstrate these ideas in Section \ref{sec:Examples} by computing the learning coefficients of a discrete mixture model which comes from a study involving 132 schizophrenic patients.

To introduce some notation, given $x \in \R^d$, let $\mathcal{A}_x(\R^d)$ be the ring of real-valued functions $f:\R^d \rightarrow \R$ that are analytic at $x$. We sometimes shorten the notation to $\mathcal{A}_x$ when it is clear that we are working with the space $\R^d$. When $x=0$, it is convenient of think of $\mathcal{A}_0$ as a subring of the formal power series ring $\R[[\omega_1, \ldots, \omega_d]] = \R[[\omega]]$. It consists of power series which are convergent in some neighborhood of the origin. For all $x$, $\mathcal{A}_x$ is isomorphic to $\mathcal{A}_0$ by translation. Given a subset $\Omega \subset \R^d$, let $\mathcal{A}_{\Omega}$ be the ring of real functions analytic at each point $x \in \Omega$. Locally, each function can be represented as a power series centered at $x$. Given $f\in\mathcal{A}_\Omega$, define the analytic variety $\mathcal{V}_{\Omega}(f)=\{\omega \in \Omega : f(\omega) = 0\}$ while for an ideal $I \subset \mathcal{A}_\Omega$, we set $\mathcal{V}_{\Omega}(I) = \cap_{f \in I} \mathcal{V}_{\Omega}(f)$. Lastly, given a finite multiset $S \subset \R$, let $\nmin S$ denote the number of times the minimum is attained in $S$.

\section{Resolution of Singularities}
\label{sec:ResSing}

In this section, we introduce Hironaka's theorem on resolutions of singularities. We derive real log canonical thresholds of monomial functions, and demonstrate how such resolutions allow us to find the thresholds of non-monomial functions. We show that the threshold of a function over a compact set is the minimum of local thresholds, and present an example where the threshold at a boundary point depend on the boundary inequalities. We discuss the problem of locating singularities with the smallest threshold, and end this section with formulas for the asymptotic expansion of a Laplace integral.

Before we explore real log canonical thresholds of ideals, let us study those of functions. Given a compact subset $\Omega$ of $\R^d$, a real analytic function $f \in \mathcal{A}_\Omega$ with $f \not\equiv 0$, and a smooth function $\varphi:\R^d \rightarrow \R$, consider the zeta function
\begin{eqnarray}
\zeta(z) = \int_{\Omega} \big|f(\omega)\big|^{-z} \,\,|\varphi(\omega)| \,d\omega,\quad z \in \C.
\label{eq:ZetafPhi}
\end{eqnarray}
This function is well-defined for $z \in \R_{\leq 0}$. If $\zeta(z)$ can be
continued analytically to the whole complex plane $\C$, then all its poles
are isolated points in $\C$. Moreover, if all its poles
are real, then there exists a smallest positive pole $\lambda$. Let
$\theta$ be the multiplicity of this pole. The pole $\lambda$ is the \emph{real log canonical threshold} of $f$ with respect to
$\varphi$ over $\Omega$. If $\zeta(z)$ has no poles, we set $\lambda
= \infty$ and leave $\theta$ undefined. Let
$\RLCT_\Omega(f;\varphi)$ be the pair $(\lambda, \theta)$. By abuse of notation, we sometimes
refer to this pair as the real log canonical threshold of $f$. We order these pairs such that $(\lambda_1, \theta_1) > (\lambda_2, \theta_2)$ if $\lambda_1 > \lambda_2$, or $\lambda_1=\lambda_2$ and $\theta_1 < \theta_2$. Intuitively, considering the asymptotics of $\log Z(N)$ in Theorem \ref{thm:Watanabe}, the ordering is defined in this way so that $(\lambda_1, \theta_1) > (\lambda_2, \theta_2)$ if and only if
$$
\lambda_1 \log N - (\theta_1 -1)\log \log N > \lambda_2 \log N - (\theta_2 -1)\log \log N
$$
for sufficiently large $N$. Lastly, let $\RLCT_\Omega \,f$  denote $\RLCT_\Omega (f; 1)$ where $1$ is~the constant unit function. 

We start with a simple class of functions for which it is easy to compute the real log canonical threshold. It is the class of monomials $\omega_1^{\kappa_1} \cdots \omega_d^{\kappa_d} = \omega^\kappa$.

\begin{prop} \label{thm:MonomialRes}
Let $\kappa = (\kappa_1, \ldots, \kappa_d)$ and $\tau = (\tau_1, \ldots, \tau_d)$ be vectors of non-negative integers. If  $\Omega$ is the positive orthant $\R^d_{\geq 0}$ and $\phi: \R^d \rightarrow \R$ is compactly supported  and smooth with $\phi(0) > 0$, then $\RLCT_\Omega(\omega^\kappa; \omega^\tau\phi) = (\lambda, \theta)$ where 
\begin{eqnarray*}
\lambda = \min_{1\leq j\leq d} \{ \frac{\tau_j+1}{\kappa_j}\}, & & \theta =\nmin_{1\leq j\leq d} \{ \frac{\tau_j+1}{\kappa_j}\}.
\end{eqnarray*}
\end{prop}
\begin{proof}
See \cite[Lemma 7.3]{AGV}. The idea is to express $\phi(\omega)$ as $T_s(\omega) + R_s(\omega)$ where $T_s$ is the $s$-th degree Taylor polynomial and $R_s$ the difference. We then integrate the main term $|f|^{-z} \,T_s$ explicitly and show that the integral of the remaining term $|f|^{-z} R_s$ does not have smaller poles. This process gives the analytic continuation of $\zeta(z)$ to the whole complex plane, so we have the Laurent expansion
\begin{eqnarray}
\label{eq:LaurentExp}
\zeta(z) =  \sum_{\alpha>0} \sum_{i=1}^{d}\frac{d_{i,\alpha}}{(z-\alpha)^i} + P(z)
\end{eqnarray}
where the poles $\alpha$ are positive rational numbers and $P(z)$ is a polynomial.
\end{proof}

For non-monomial $f(\omega)$, Hironaka's celebrated theorem \cite{Hi} on the \emph{resolution of singularities} tells us that we can always reduce to the monomial case. Here, a $d$-dimensional real analytic manifold is a topological space (second countable and Hausdorff) that can be covered by charts which are homeomorphic to open balls in $\R^d$ and where the transition maps between charts are real analytic maps.

\begin{thm}[Resolution of Singularities] \label{thm:ResOfSing}
Let $f$ be a non-constant real analytic function in some neighborhood $\Omega \subset \R^d$ of the origin with $f(0)=0$. Then, there exists a triple $(M, W, \rho)$ where
\begin{enumerate}
\item[a.] $W \subset \Omega$ is a neighborhood of the origin,
\item[b.] $M$ is a $d$-dimensional real analytic manifold,
\item[c.] $\rho:M \rightarrow W$ is a real analytic map
\end{enumerate}
satisfying the following properties.
\begin{enumerate}
\item[i.] $\rho$ is proper, i.e. the inverse image of any compact set is compact.
\item[ii.] $\rho$ is a real analytic isomorphism between $M\setminus \mathcal{V}_{M}(f \circ \rho)$ and $W\setminus \mathcal{V}_{W}(f)$. 
\item[iii.] For any $y \in \mathcal{V}_{M}(f\circ \rho)$, there exists a local chart $M_y$ with coordinates $\mu = (\mu_1, \mu_2, \ldots \mu_d)$ such that $y$ is the origin and 
$$
f\circ \rho(\mu) = a(\mu) \mu_1^{\kappa_1} \mu_2^{\kappa_2} \cdots \mu_d^{\kappa_d} = a(\mu)\mu^\kappa
$$
where $\kappa_1, \kappa_2, \ldots, \kappa_d$ are non-negative integers and $a$ is a real analytic function with $a(\mu) \neq 0$ for all $\mu$. Furthermore, the Jacobian determinant equals
$$
|\rho'(\mu)| = h(\mu) \mu_1^{\tau_1} \mu_2^{\tau_2} \cdots \mu_d^{\tau_d} = h(\mu) \mu^\tau
$$ 
where $\tau_1, \tau_2, \ldots, \tau_d$ are non-negative integers and $h$ is a real analytic function with $h(\mu) \neq 0$ for all $\mu$.
\end{enumerate}
\end{thm}

We say that $(M, W, \rho)$ is a \emph{resolution of singularities} or a \emph{desingularization} of $f$ at the origin. The set of points in $M$ where $\rho$ is not one-to-one is the \emph{excep-} \emph{tional divisor}. From properties (i) and (ii), it also follows that $\rho$ is surjective: if $x \in \mathcal{V}_{W}(f)$, we pick a compact neighborhood $V$ of $x$ and a sequence $x_1, x_2, \ldots$ of points in $V \setminus \mathcal{V}_{W}(f)$ converging to $x$. The sequence can be chosen off the variety because the variety has measure zero. Then, the preimages $\rho^{-1}(x_1), \rho^{-1}(x_2), \ldots$ contain a converging subsequence with limit $y$, and $\rho(y)=x$ by continuity.

Now, let us desingularize a list of functions simultaneously.

\begin{cor}[Simultaneous Resolutions] \label{thm:SimultRes}
Let $f_1,\ldots, f_l$ be non-constant real analytic functions in some neighborhood $\Omega \subset \R^d$ of the origin with all $f_i(0)=0$. Then, there exists a triple $(M, W, \rho)$ that desingularizes each $f_i$ at the origin.
\end{cor}
\begin{proof} The idea is to desingularize the product $f_1(\omega) \cdots f_l(\omega)$ and to show that such a resolution of singularities is also a resolution for each $f_i$. See \cite[Thm 11]{W} and \cite[Lemma 2.3]{G2} for details. 
\end{proof}

For the rest of this section, let $\Omega = \{\omega \in \R^d, g_1(\omega) \geq 0, \ldots, g_l(\omega) \geq 0\}$ be compact and semianalytic. We also assume that $f, \varphi \in \mathcal{A}_\Omega$, and that $f, g_1, \ldots, g_l$ are not constant functions.

\begin{lem}\label{thm:LocalRLCT}
For each $x \in \Omega$, there is a neighborhood $\Omega_x$ of $x$ in $\Omega$ such that 
for all smooth functions $\phi$ on $\Omega_x$ with $\phi(x) > 0$,
$$
\RLCT_{\Omega_x}(f; \varphi \phi) = \RLCT_{\Omega_x}(f;\varphi).
$$
\end{lem}
\begin{proof}
Let $x \in \Omega$. If $f(x) \neq 0$, then by the continuity of $f$, there exists a small neighborhood $\Omega_x$ where $0 < c_1 < |f(\omega)| < c_2$ for some constants $c_1, c_2$. Hence, for all smooth functions $\phi$, the zeta functions
$$
\int_{\Omega_x} \big|f(\omega)\big|^{-z} |\varphi(\omega) \phi(\omega)| \,d\omega \quad \mbox{and} \quad \int_{\Omega_x} \big|f(\omega) \big|^{-z}  |\varphi(\omega)| \,d\omega
$$
do not have any poles, so the lemma follows in this case.

Suppose $f(x) = 0$. By Corollary \ref{thm:SimultRes}, we have a simultaneous local resolution of singularities $(M, W,\rho)$ for the functions $f, \varphi, g_1, \ldots, g_l$ vanishing at $x$. For each point $y$ in the fiber $\rho^{-1}(x)$, we have a local chart $M_{y}$ satisfying property (iii) of Theorem \ref{thm:ResOfSing}. Since $\rho$ is proper, the fiber $\rho^{-1}(x)$ is compact so there is a finite subcover $\{M_{y}\}$. We claim that the image $\rho(\bigcup M_{y})$ contains a neighborhood $W_x$ of $x$ in $\R^d$. Indeed, otherwise, there exists a bounded sequence $\{x_1, x_2, \ldots \}$ of points in $W \setminus \rho(\bigcup M_{y})$ whose limit is $x$. We pick a sequence $\{y_1, y_2, \ldots \}$ where $\rho(y_i) = x_i$. Since the $x_i$ are bounded, the $y_i$ lie in a compact set so there is a convergent subsequence $\{\tilde{y}_i\}$ with limit $y_*$. The $\tilde{y}_i$ are not in the open set $\bigcup M_{y}$ so nor is $y_*$. But $\rho(y_*) = \lim \rho(\tilde{y}_i) = x$ so $y_* \in \rho^{-1}(x) \subset M_{y}$, a contradiction.

Now, define $\Omega_x = W_x \cap \Omega$ and let $\{\mathcal{M}_y\}$ be the collection of all sets $\mathcal{M}_y = M_y \cap \rho^{-1}(\Omega_x)$ which have positive measure. Picking a partition of unity $\{\sigma_y(\mu) \}$ subordinate to $\{\mathcal{M}_y\}$ such that $\sigma_y$ is positive at $y$ for each $y$ \cite[Theorem 6.5]{W}, we write the zeta function $\zeta(z) =\int_{\Omega_x} |f(\omega)|^{-z} |\varphi(\omega) \phi(\omega)|\, d\omega$ as
$$
\sum_y \int_{\mathcal{M}_y}\big| f\circ \rho(\mu)\big|^{-z} \,|\varphi\circ\rho(\mu) || \phi\circ\rho(\mu)||\rho'(\mu) | \sigma_y(\mu) \,d\mu.
$$
For each $y$, the boundary conditions $g_i\circ \rho(\mu) \geq 0$ become monomial inequalities, so $\mathcal{M}_y$ is the union of closed orthant neighborhoods of $y$. The integral over $\mathcal{M}_y$ is then the sum of integrals of the form 
$$
\zeta_y(z) = \int_{\R^d_{\geq 0}} \mu^{-\kappa z+\tau} \psi(\mu) d\mu
$$
where $\kappa$ and $\tau$ are non-negative integer vectors while $\psi$ is a compactly supported smooth function with $\psi(0) > 0$. Note that $\kappa$ and $\tau$ do not depend on $\phi$ nor on the choice of orthant at $y$. By Proposition \ref{thm:MonomialRes}, the smallest pole of $\zeta_y(z)$ is
\begin{eqnarray*}
\lambda_y = \min_{1\leq j\leq d} \{ \frac{\tau_j+1}{\kappa_j}\}, \quad
\theta_y = \nmin_{1\leq j\leq d} \{ \frac{\tau_j+1}{\kappa_j}\}.
\end{eqnarray*}
Now,
$
\RLCT_{\Omega_x} (f;\varphi \phi) = \min_y \{(\lambda_y, \theta_y)\}.
$
Since this formula is independent of $\phi$, we set $\phi=1$ and the lemma follows.
\end{proof}
\begin{prop} \label{thm:GlobalRLCT}
Let $\phi:\Omega \rightarrow \R$ be positive and smooth. Then, for sufficiently small neighborhoods $\Omega_x$, the set $\{\RLCT_{\Omega_x}(f; \varphi) : x \in \Omega\}$ has a minimum and
$$\RLCT_\Omega(f; \varphi \phi ) = \min_{x \in \Omega} \,\RLCT_{\Omega_x} (f; \varphi).$$
\end{prop}
\begin{proof}
Lemma \ref{thm:LocalRLCT} associates a small neighborhood to each point in the compact set $\Omega$, so there exists a finite subcover $\{\Omega_{x}:x \in S\}$. Let $\{\sigma_{x}(\omega)\}$ be a smooth partition of unity subordinate to this subcover where $\sigma_{x}(x)>0$ for all $x$. Then,
$$\int_{\Omega} \big| f(\Omega) \big|^{-z} |\varphi(\omega) \phi(\omega) | \,d\omega = \sum_{x \in S} \int_{\Omega_x} \big| f(\Omega) \big|^{-z} | \varphi(\omega) \phi(\omega)|\, \sigma_{x}(\omega)\, d\omega.
$$
From this finite sum, we have
$$
\RLCT_\Omega(f; \varphi \phi) = \min_{x \in S}\,\RLCT_{\Omega_{x}} (f; \varphi \phi \sigma_{x}) = \min_{x \in S}\,\RLCT_{\Omega_{x}} (f; \varphi).
$$
Now, if $y \in \Omega \setminus S$, let $\Omega_y$ be a neighborhood of $y$ prescribed by Lemma \ref{thm:LocalRLCT} and consider the cover $\{\Omega_x : x \in S\} \cup \{\Omega_y\}$ of $\Omega$. After choosing a partition of unity subordinate to this cover and repeating the above argument, we get 
$$
\RLCT_\Omega(f; \varphi \phi) \leq \RLCT_{\Omega_{y}} (f; \varphi) \quad \mbox{ for all  } y\in \Omega.
$$
Combining the two previously displayed equations proves the proposition.
\end{proof}

Abusing notation, we now let $\RLCT_{\Omega_x} (f; \varphi)$ represent the real log canonical threshold for a sufficiently small neighborhood $\Omega_x$ of $x$ in $\Omega$. If $x$ is an interior point of $\Omega$, we denote the threshold at $x$ by $\RLCT_{x} (f; \varphi)$. 

\begin{cor}[See also {\cite[\S 4.5]{W}}]
\label{thm:AnalyticCont}
Given a compact semianalytic set $\Omega \subset \R^d$, a nearly analytic function $\varphi:\Omega \rightarrow \R$, and $f \in \mathcal{A}_{\Omega}$ satisfying $f(x)=0$ for some $x \in \Omega$, the zeta function (\ref{eq:ZetafPhi}) can be continued analytically to $\C$. It has a Laurent expansion (\ref{eq:LaurentExp}) whose poles are positive rational numbers with a smallest element.
\end{cor}
\begin{proof}
The proofs of Lemma \ref{thm:LocalRLCT} and Proposition \ref{thm:GlobalRLCT} outline a way to compute the Laurent expansion of the zeta function (\ref{eq:ZetafPhi}).
\end{proof}

\begin{remark} \label{rem:PriorAbsoluteValue}
In our definition of real log canonical thresholds, we considered integrals with respect to densities $|\varphi(\omega)|\,d\omega$ for some nearly analytic function $\varphi$, while Watanabe only considers the special case where the density is $\varphi(\omega)\,d\omega$ for some smooth positive function $\varphi$. Our general case includes the situation where the absolute value of a Jacobian determinant is multiplied to the density under a change of variables. To prove the basic properties of real log canonical thresholds, we need to resolve the singularities of the variety $\varphi=0$ together~with those cut out by $f, g_1, \ldots, g_l$, as demonstrated in Lemma~\ref{thm:LocalRLCT}. 
\end{remark}

\begin{example}\label{ex:RLCTdependsonboundary}
We now show that the threshold at a boundary point depends on the boundary inequalities. Consider the following two small neighborhoods of the origin in some larger compact set.
$$
\begin{array}{c}
\Omega_1 = \{ (x,y) \in \R^2 : 0 \leq x \leq y  \leq \varepsilon\}\\
\Omega_2 = \{ (x,y) \in \R^2 : 0 \leq y \leq x  \leq \varepsilon\}
\end{array}
$$
To compute the real log canonical threshold of the function $xy^2$ over these sets, we have the corresponding zeta functions below.
$$
\begin{array}{rclcl}
\vspace{0.05in} \zeta_1(z) &=& \displaystyle \int_0^\varepsilon \int_0^y x^{-z} y^{-2z} \,dx\,dy &=& \displaystyle \frac{\varepsilon^{-3z+2}}{(-z+1)(-3z+2)} \\
\zeta_2(z) &=& \displaystyle \int_0^\varepsilon \int_0^x x^{-z} y^{-2z} \,dy\,dx &=& \displaystyle \frac{\varepsilon^{-3z+2}}{(-2z+1)(-3z+2)} 
\end{array}
$$
This shows that $\RLCT_{\Omega_1}(xy^2) = 2/3$ while $\RLCT_{\Omega_2}(xy^2) = 1/2$. \qed
\end{example}

Because the real log canonical threshold over a set $\Omega \subset \R^d$ is the minimum of thresholds at points $x \in \Omega$, we want to know where this minimum is achieved. Let us study this problem topologically. Consider a  locally finite collection $\mathcal{S}$~of pairwise disjoint submanifolds $S \subset \Omega$ such that $\Omega = \cup_{S \in \mathcal{S}} S$ and each $S$ is locally closed, i.e. the intersection of an open and a closed subset. Let $\overline{S}$ be the closure of $S$. We say $\mathcal{S}$ is a \emph{stratification} of $\Omega$ if $S \cap \overline{T} \neq \emptyset$ implies $S \subset \overline{T}$ for all $S,T \in \mathcal{S}$.
A stratification $\mathcal{S}$ of $\Omega$ is a \emph{refinement} of another stratification $\mathcal{T}$ if $S \cap T \neq \emptyset$ implies $S \subset T$ for all $S \in \mathcal{S}$ and $T \in \mathcal{T}$.

Let the amplitude $\varphi:\Omega \rightarrow \R$ be nearly analytic. Let $S_{(\lambda,
\theta), 1}, \ldots, S_{(\lambda,
\theta), r}$ be the connected components of the set $\{x \in \Omega : \RLCT_{\Omega_x}(f;\varphi) = (\lambda,\theta)\}$, and let $\mathcal{S}$ denote the collection $\{S_{(\lambda,\theta), i}\}$ where we vary over all $\lambda$, $\theta$ and $i$. Now, define the order $\ord_xf$ to be the smallest degree of a monomial appearing in a series expansion of $f$ at $x \in \Omega$ \cite[\S 3.9]{E}. This number is independent of the choice of local coordinates $\omega_1, \ldots, \omega_d$ because it is the largest integer $k$ such that $f \in \mathfrak{m}_x^k$ where $\mathfrak{m}_x=\{g \in A_x : g(x) = 0\}$ is the vanishing ideal of $x$. Define $T_{l, 1}, \ldots, T_{l,s}$ to be the connected components of the set $\{x \in \Omega: \ord_xf = l\}$ and let $\mathcal{T}$ be the collection $\{T_{l,j}\}$ where we vary over all $l$ and $j$. We conjecture the following relationship between $\mathcal{S}$ and $\mathcal{T}$. It implies that the minimum real log canonical threshold over a set must occur at a point of highest order.

\begin{conj}\label{conj:stratification}
The collections $\mathcal{S}$ and $\mathcal{T}$ are stratifications of $\Omega$. Furthermore, if the amplitude $\varphi$ is a positive smooth function, then $\mathcal{S}$ refines $\mathcal{T}$.
\end{conj}

\bigskip

Laplace integrals such as (\ref{eq:Marginal}) occur frequently in physics, statistics and other applications. At first, the relationship between their asymptotic expansions and the zeta function (\ref{eq:WatanabeZeta}) seems strange. The key is to write these integrals as
\begin{eqnarray*} 
&Z(N) =  \displaystyle \int_\Omega e^{-N|f(\omega)|} |\varphi(\omega)| \,d\omega = \int_0^{\infty} e^{-N t} v(t) \,dt &\\
&\zeta(z) = \displaystyle  \int_\Omega \big|f(\omega)\big|^{-z} |\varphi(\omega)| \,d\omega = \int_0^\infty t^{-z} v(t) \,dt&
\end{eqnarray*}
where $v(t)$ is the state density function \cite{W} or Gelfand-Leray function \cite{AGV}
\begin{eqnarray*}
v(t) = \frac{d}{dt} \int_{0<|f(\omega)| < t} |\varphi(\omega)|\, d\omega.
\end{eqnarray*}
Formally, $Z(N)$ is the Laplace transform of $v(t)$ while $\zeta(z)$ is its Mellin transform. Note that contrary to its name, $v(t)$ is not strictly a function, but it can be defined as a Schwartz distribution. Next, we study the series expansions
\begin{eqnarray}
\label{eq:SeriesZ} Z(N) &\approx&\sum_\alpha \sum_{i=1}^{d} c_{\alpha,i} N^{-\alpha} (\log N)^{i-1} \\
\label{eq:Seriesv} v(t) &\approx& \sum_{\alpha} \sum_{i=1}^{d} b_{\alpha,i}\, t^{\alpha}(\log t)^{i-1}\\
\label{eq:SeriesZeta} \zeta(z) &\sim& \sum_{\alpha} \sum_{i=1}^{d}d_{\alpha,i}(z-\alpha)^{-i} 
\end{eqnarray}
where (\ref{eq:SeriesZ}) and (\ref{eq:Seriesv}) are asymptotic expansions while (\ref{eq:SeriesZeta}) is the principal part of the Laurent series expansion. Here, the number $d$ of summands is the dimension of the parameter space $\Omega \subset \R^d$. Formulas relating the coefficients $b_{\alpha,i}, c_{\alpha,i}$ and $d_{\alpha,i}$ are then deduced from the Laplace and Mellin transforms of $t^{\alpha}(\log t)^i$. For more detailed expositions on this subject, we refer the reader to Arnol'd--Guse\u\i n-Zade--Varchenko \cite[\S6-7]{AGV}, Watanabe \cite[\S4]{W} and Greenblatt \cite{G}. 

Using this strategy, we now give explicit formulas for the asymptotic expansion of an arbitrary Laplace integral. Our formulas generalize those of Arnol'd--Guse\u\i n-Zade--Varchenko \cite[\S6-7]{AGV} because they apply also to parameter spaces $\Omega$ with analytic boundary. Watanabe \cite[Remark 4.5]{W} gives a similar asymptotic expansion for bounded parameter spaces but we derive precise relationships between the asymptotic coefficients $c_{\alpha, i}$ and the Laurent coefficients $d_{\alpha,i}$ in terms of derivatives $\Gamma^{(i)}$ of Gamma functions.

\begin{thm}
\label{thm:AsympExp}
Let $\Omega \subset \R^d$ be a compact semianalytic subset and $\varphi:\Omega \rightarrow \R$ be nearly analytic. If $f \in \mathcal{A}_\Omega$ with $f(x)=0$ for some $x \in \Omega$, the Laplace integral 
$$
Z(N) = \int_\Omega e^{-N|f(\omega)|} |\varphi(\omega)| \,d\omega
$$
has the asymptotic expansion 
\begin{equation}
\sum_{\alpha}\sum_{i=1}^{d} c_{\alpha, i} \,N^{-\alpha} (\log N)^{i-1}.
\end{equation}
The $\alpha$ in this expansion range over positive rational numbers which are poles of
\begin{equation}
\zeta(z) =  \int_{\Omega_\delta} \big|f(\omega)\big|^{-z} |\varphi(\omega)| \,d\omega
\end{equation}
for any $\delta > 0$ and $\Omega_\delta = \{\omega \in \Omega : |f(\omega)| < \delta \}$. The coefficients $c_{\alpha,i}$ satisfy
\begin{equation}
\label{eq:ciAlpha}
c_{\alpha,i}= \frac{(-1)^{i}}{(i-1)! } \sum_{j=i}^{d} \frac{\Gamma^{(j-i)}(\alpha)}{(j-i)!} \,d_{\alpha,j} 
\end{equation}
where $d_{\alpha,j}$ is the coefficient of $(z-\alpha)^{-j}$ in the Laurent expansion of $\zeta(z)$.
\end{thm}
\begin{proof}
First, set $\delta = 1$. We split the integral $Z(N)$ into two parts:
$$
Z(N) = \int_{|f(\omega)| < 1} e^{- N|f(\omega)|} |\varphi(\omega)|\, d\omega + \int_{|f(\omega)| \geq 1} e^{- N|f(\omega)|} |\varphi(\omega)| \,d\omega.
$$
The second integral is bounded above by $Ce^{-N}$ for some non-negative constant $C$, so asymptotically it goes to zero more quickly than any $N^{-\alpha}$. For the first integral, we write $\zeta(z)$ as the Mellin transform of the state density function $v(t)$.
$$
\zeta(z) = \int_{|f(\omega)| < 1} \big|f(\omega)\big|^{-z}|\varphi(\omega)| \,d\omega = \int_0^1 t^{-z} v(t) \,dt.\\
$$
By Corollary \ref{thm:AnalyticCont}, $\zeta(z)$ has a Laurent expansion (\ref{eq:LaurentExp}). Since $|f(\omega)| < 1$, by~domin-ated convergence $\zeta(z) \to 0$ as $z \to -\infty$, so the polynomial part $P(z)$ is identically zero.  Applying the inverse Mellin transform \cite{BBO} to $\zeta(z)$, we get a series expansion (\ref{eq:Seriesv}) of the state density function $v(t)$. Applying the Laplace transform to $v(t)$ in turn gives the asymptotic expansion (\ref{eq:SeriesZ}) of $Z(N)$. The formulas
\begin{align*}
\displaystyle \int_0^\infty e^{-N t} \,t^{\alpha-1} (\log t)^i \,dt &\approx \sum_{j=0}^{i} \binom{i}{j} (-1)^j \Gamma^{(i-j)}(\alpha)\,N^{-\alpha}(\log N)^j\\
\int_0^{1} t^{-z} \,t^{\alpha-1} (\log t)^i \,dt &= -\, i!\, (z-\alpha)^{-(i+1)} 
\end{align*}
from \cite[Thm 7.4]{AGV} and \cite[Ex 4.7]{W} give us the relations
$$
c_{\alpha,i} = (-1)^{i-1} \sum_{j=i}^{d} \binom{j-1}{i-1} \Gamma^{(j-i)}(\alpha) \,b_{\alpha-1,j}, \quad d_{\alpha,j} = - \,(j-1)! \,b_{\alpha-1,j}.
$$
Equation (\ref{eq:ciAlpha}) follows immediately. Finally, for all other values of $\delta$, we write
$$
\int_\Omega \big|f(\omega)\big|^{-z} |\varphi(\omega)| d\omega = \int_{\Omega_\delta} \big|f(\omega)\big|^{-z} |\varphi(\omega)| d\omega + \int_{|f(\omega)| \geq \delta} \big|f(\omega)\big|^{-z} |\varphi(\omega)| d\omega.
$$
The last integral does not have any poles, so the principal parts of the Laurent expansions of the first two integrals are the same for all $\delta$.
\end{proof}

\section{Real Log Canonical Thresholds} 
\label{sec:RLCT}

In this section, we prove fundamental properties of real log canonical thresholds (RLCTs) which will allow us to calculate these thresholds more efficiently. The learning coefficient of a statistical model is shown to be the RLCT of the ideal generated by its defining equations. 

In this section, let $\Omega \subset \R^d$ be a compact semianalytic subset and let $\varphi:\Omega \rightarrow \R$ be nearly analytic. Given functions $f_1, \ldots, f_r \in \mathcal{A}_\Omega$, let $\RLCT_\Omega (f_1, \ldots, f_r; \varphi)$ be the smallest pole and multiplicity of the zeta function \eqref{eq:IdealZetaFn}. Recall that these pairs are ordered by the rule $(\lambda_1, \theta_1) > (\lambda_2, \theta_2)$ if $\lambda_1 > \lambda_2$, or $\lambda_1=\lambda_2$ and $\theta_1 < \theta_2$. For $x \in \Omega$, we define $\RLCT_{\Omega_x} (f_1, \ldots, f_r; \varphi)$ to be the threshold for a sufficiently small neighborhood $\Omega_x$ of $x$ in $\Omega$. 

\begin{remark}
\label{rem:lct}
The (complex) log canonical threshold may be defined in a similar fashion. It is the smallest pole of the zeta function 
$$
\zeta(z) = \int_{\Omega} \Big(|f_1(\omega)|^2+\cdots+|f_r(\omega)|^2\Big)^{-z} d\omega.
$$
Note that the $f_i^2$ have been replaced by $|f_i|^2$ and the exponent $-z/2$ is changed to $-z$. Crudely, this factor of $2$ comes from the fact that $\C^d$ is a real vector space of dimension $2d$. The complex threshold is often different from the RLCT~\cite{S}. From the algebraic geometry point of view, more is known about complex log canonical thresholds than about real log canonical thresholds. Many results in this paper were motivated by their complex analogs \cite{BL, H, K, L}. 

\end{remark}

Now, we give several equivalent definitions of $\RLCT_\Omega(f_1, \ldots, f_r; \varphi)$ which are helpful in proofs of the fundamental properties.
\begin{prop} \label{thm:DefRLCT}
Given functions $f_1, \ldots, f_r \in \mathcal{A}_\Omega$ such that each $f_i \not\equiv 0$ and $\mathcal{V}_\Omega(\langle f_1, \ldots, f_r \rangle)$ is nonempty, the pairs $(\lambda, \theta)$ defined below are all equal.
\begin{enumerate}
\item[a.] The logarithmic Laplace integral 
$$
\displaystyle \log Z(N) = \log \int_\Omega \exp\Big(\mi N \sum_{i=1}^{r} f_i(\omega)^2\Big) |\varphi(\omega)| 
\,d\omega
$$ 
is asymptotically $-\frac{\lambda}{2} \log N + (\theta-1) \log \log N + O(1)$.
\item[b.] The zeta function 
$$
\displaystyle \zeta(z) = \int_\Omega \Big(\sum_{i=1}^{r} f_i(\omega)^2\Big)^{-z/2}  |\varphi(\omega)| \,d\omega
$$
has a smallest pole $\lambda$ of multiplicity $\theta$.
\item[c.] The pair $(\lambda, \theta)$ is the minimum
$$
\min_{x \in \Omega} \,\RLCT_{\Omega_x}(f_1, \ldots, f_r; \varphi).
$$
In fact, it is enough to vary $x$ over $\mathcal{V}_{\Omega}(\langle f_1, \ldots, f_r \rangle)$.
\end{enumerate}
\end{prop}
\begin{proof} Item (b) is the original definition of the RLCT. The equivalence of (a) and (b) follows from Theorem \ref{thm:AsympExp}, and that of (b) and (c) from Proposition \ref{thm:GlobalRLCT}. The last statement of (c) follows from the fact that the RLCT is $\infty$ for points $x \notin \mathcal{V}_{\Omega}(\langle f_1, \ldots, f_r \rangle)$. See also \cite[Thm 7.1]{W}.
\end{proof}

Our first property describes the effect of the boundary on the RLCT. 

\begin{prop} \label{thm:EffectOfBoundary} Let $x$ be a boundary point of $\Omega \subset \R^d$. Then, for every~neighborhood $W$ of $x$ in $\R^d$,
$$
\RLCT_{W} (f;\varphi) \leq \RLCT_{\Omega_x}(f; \varphi).
$$
\end{prop}
\begin{proof}
For a sufficiently small neighborhood $\Omega_x$ of $x$ in $\Omega$, we have $\Omega_x \subset W$, so the corresponding Laplace integrals satisfy $Z_{\Omega_x}(N) \leq Z_W(N)$. By Proposition \ref{thm:DefRLCT}, this gives the opposite inequality on the RLCTs.
\end{proof}

If the function whose $\RLCT$ we are finding is complicated, we may replace it with a simpler function that bounds it. Given $f, g \in \mathcal{A}_\Omega$, we say that $f$ and $g$ are \emph{equivalent} in $\Omega$ if $c_1f \leq g \leq c_2f$ in $\Omega$ for some $c_1,c_2>0$.

\begin{prop}[\!\!{\cite[Remark 7.2]{W}}] \label{thm:AsympUpBound}
Given $f, g \in \mathcal{A}_\Omega$, suppose that $0 \leq cf \leq g$ in $\Omega$ for some $c>0$. Then, $\RLCT_{\Omega} (f;\varphi) \leq \RLCT_{\Omega} (g; \varphi)$.
\end{prop}

\begin{cor} \label{thm:AsympBound}
If $f, g$ are equivalent in $\Omega$, then $\RLCT_{\Omega} (f; \varphi) = \RLCT_{\Omega} (g; \varphi)$.
\end{cor}

$\RLCT_{\Omega}(f_1^2+\cdots+f_r^2; \varphi)=(\lambda, \theta)$ implies $\RLCT_{\Omega}(f_1, \ldots, f_r; \varphi) = (2\lambda, \theta)$. From this, it seems that we should restrict ourselves to RLCTs of single and~not multiple functions. However, as the next proposition shows, multiple functions are important because they allow us to work with ideals for which different generating sets can be chosen. This gives us freedom to switch between single and multiple functions in powerful ways. For instance, special cases of this proposition such as Lemmas 3 and 4 of \cite{AW} have been used to simplify computations.

\begin{prop} \label{thm:SameIdeal}
If two sets $\{ f_1, \ldots, f_r\}$ and $ \{ g_1, \ldots, g_s\}$ of functions generate the same ideal $I \subset \mathcal{A}_\Omega$, then $$\RLCT_\Omega(f_1,\ldots,f_r; \varphi) = \RLCT_\Omega(g_1, \ldots, g_s; \varphi).$$ Define this pair to be $\RLCT_\Omega (I; \varphi)$.
\end{prop}
\begin{proof}
Each $g_j$ can be written as a combination $h_1f_1 + \cdots + h_rf_r$ of the $f_i$ where the $h_i$ are real analytic over $\Omega$. By the Cauchy-Schwarz inequality,
$$
g_j^2 \leq \big(h_1^2+\cdots+h_r^2)\big(f_1^2 + \cdots + f_r^2\big).
$$
Because $\Omega$ is compact, the $h_i$ are bounded. Thus, summing over all the $g_j$, there is some constant $c > 0$ such that, $$
\sum_{j=1}^{s} g_j^2 \leq c \sum_{i=1}^{r} f_i^2.
$$
By Proposition \ref{thm:AsympUpBound}, $\RLCT_\Omega(g_1, \ldots, g_r; \varphi) \leq \RLCT_\Omega(f_1, \ldots, f_r; \varphi)$ and by symmetry, the reverse is also true, so we are done. See also \cite[\S 2.6]{S}.
\end{proof}

For the next result, let $f_1, \ldots, f_r \in \mathcal{A}_X$ and $g_1, \ldots, g_s \in \mathcal{A}_Y$ where $X \subset \R^m$ and $Y \subset \R^n$ are compact semianalytic subsets. This occurs, for instance, when the $f_i$ and $g_j$ are polynomials with disjoint sets of indeterminates $\{x_1, \ldots, x_m\}$ and $\{y_1, \ldots, y_n\}$. Let $\varphi_X:X \rightarrow \R$ and $\varphi_Y:Y\rightarrow \R$ be nearly analytic. Define $(\lambda_X, \theta_X) = \RLCT_X(f_1, \ldots, f_r; \varphi_X)$ and $(\lambda_Y, \theta_Y) = \RLCT_Y(g_1, \ldots, g_s; \varphi_Y)$.

By composing with projections $X{\times}Y\rightarrow X$ and $X{\times}Y\rightarrow Y$, we may regard the $f_i$ and $g_j$ as functions analytic over $X{\times}Y$. Let $I_X$ and $I_Y$ be ideals in~$\mathcal{A}_{X{\times}Y}$ generated by the $f_i$ and $g_j$ respectively. Recall that the sum $I_X+I_Y$ is generated by all the $f_i$ and $g_j$ while the product $I_XI_Y$ is generated by $f_ig_j$ for all $i,j$.

\begin{prop}
\label{thm:DisjointVars}
The RLCTs for the sum and product of ideals $I_X$ and $I_Y$ are
\begin{eqnarray*}
\RLCT_{X{\times}Y}(I_X+I_Y; \varphi_X\varphi_Y) &=& (\lambda_X+\lambda_Y,\,\, \theta_X+\theta_Y-1), \\
\RLCT_{X{\times}Y}(I_XI_Y; \varphi_X\varphi_Y) &=& \left\{ 
\begin{array}{ll}
(\lambda_X, \,\,\theta_X) & \mbox{if } \lambda_X < \lambda_Y, \\
(\lambda_Y, \,\,\theta_Y) & \mbox{if } \lambda_X > \lambda_Y, \\
(\lambda_X, \,\,\theta_X+\theta_Y) & \mbox{if } \lambda_X = \lambda_Y.
\end{array}
\right. 
\end{eqnarray*}
\end{prop}
\begin{proof} Define $f(x) = f_1^2 +\cdots +f_r^2$ and $g(y)=g_1^2+\cdots+g_s^2$, and let $Z_X(N)$ and $Z_Y(N)$ be the corresponding Laplace integrals. By Proposition \ref{thm:DefRLCT},
\begin{eqnarray*}
\log Z_X(N) &=&\textstyle  -\frac{1}{2} \lambda_X\log N + (\theta_X-1) \log \log N + O(1)\\
\log Z_Y(N) &=&\textstyle -\frac{1}{2} \lambda_Y\log N + (\theta_Y-1) \log \log N + O(1)
\end{eqnarray*}
asymptotically. If $(\lambda, \theta) = \RLCT_{X{\times}Y}(I_X+I_Y; \varphi_X\varphi_Y)$, then
\begin{eqnarray*}
&&\!\!\!\!\!\!\!\!\!\!\!\!\!\!\! -\textstyle \frac{1}{2}\lambda \log N + (\theta-1) \log \log N + O(1) \\
&&= \log \textstyle\, \int_{X{\times}Y} e^{- N f(x) -N g(y)} |\varphi_X||\varphi_Y|\,dx\,dy \\
&&=\textstyle \log \big(\int_{X} e^{-N f(x)} |\varphi_X| \,dx\big) \big(\int_{Y} e^{-N g(y)} |\varphi_Y| \,dy\big)  \\
&&= \log Z_X(N) + \log Z_Y(N) \\
&&= \textstyle -\frac{1}{2}(\lambda_X+\lambda_Y) \log N + (\theta_X+\theta_Y-2) \log \log N + O(1) 
\end{eqnarray*}
and the first result follows. For the second result, note that 
\begin{eqnarray*}
f(x)g(y) &=& f_1^2g_1^2 + f_1^2g_2^2 + \cdots + f_r^2g_s^2.
\end{eqnarray*}
Let $\zeta_X(z)$ and $\zeta_Y(z)$ be the zeta functions corresponding to $f(x)$ and $g(y)$. By Proposition \ref{thm:DefRLCT}, $(\lambda_X, \theta_X)$ and $(\lambda_Y, \theta_Y)$ are the smallest poles of $\zeta_X(z)$ and $\zeta_Y(z)$ while $\RLCT_{X{\times}Y} (I_XI_Y; \varphi_X\varphi_Y)$ is the smallest pole of
\begin{eqnarray*}
  \zeta(z) &=& \textstyle\int_{X{\times}Y} \big(f(x)g(y)\big)^{-z/2} |\varphi_X||\varphi_Y| \,dx\,dy \\ &=&  \textstyle\big(\int_{X} f(x)^{-z/2} |\varphi_X| \,dx\big)\big(\int_{Y} g(y)^{-z/2} |\varphi_Y| \,dy\big) \,\,\,\,=\,\,\,\, \zeta_X(z) \zeta_Y(z). 
\end{eqnarray*}
The second result then follows from the relationship between the poles.
\end{proof}

Our last property tells us the behavior of RLCTs under a change of variables. Consider an ideal $I \subset \mathcal{A}_W$ where $W$ is a neighborhood of the origin. Let $M$ be a real analytic manifold and $\rho:M \rightarrow W$ be a proper real analytic map. Then, the \emph{pullback} $\rho^* I$ is locally the ideal of real analytic functions on $M$ that is generated by $f \circ \rho$ for all $f \in I$ (also called the inverse image ideal sheaf \cite[\S 3.3]{K2}). If $\rho$ is an isomorphism between $M \setminus \mathcal{V}(\rho^*I)$ and $W\setminus \mathcal{V}(I)$, we say that $\rho$ is a \emph{change of variables away from} $\mathcal{V}(I)$. Let $|\rho'|$ denote the Jacobian determinant of $\rho$. We call $(\rho^*I; (\varphi \circ \rho)  |\rho'|)$ the \emph{pullback pair}.
\begin{prop}
\label{thm:ChangeOfVar}
Let  $W$ be a neighborhood of the origin and $I \subset \mathcal{A}_W$ a finitely generated ideal. If $M$ is a real analytic manifold, $\rho:M \rightarrow W$ is a change of variables away from $\mathcal{V}(I)$ and $\mathcal{M} = \rho^{-1}(\Omega \cap W)$, then
$$\RLCT_{\Omega_0}(I;\varphi) =  \min_{x \in \rho^{-1}(0)} \RLCT_{\mathcal{M}_x}(\rho^*I; (\varphi \circ \rho) |\rho'|).$$
\end{prop}
\begin{proof}
Let $f_1, \ldots, f_r$ generate $I$ and let $f=f_1^2+\cdots+f_r^2$. Then, $\RLCT_{\Omega_0}(I;\varphi)$ is the smallest pole and multiplicity of the zeta function
$$
\zeta(z) = \int_{\Omega_0} f(\omega)^{-z/2}|\varphi(\omega)| \,d\omega
$$
where $\Omega_0 \subset W$ is a sufficiently small neighborhood of the origin in $\Omega$. Applying the change of variables $\rho$, we have
$$
\zeta(z) = \int_{\rho^{-1}(\Omega_0)} f\circ \rho(\mu)^{-z/2}|\varphi\circ \rho(\mu)| |\rho'(\mu)| \, d\mu.
$$
The proof of Lemma \ref{thm:LocalRLCT} shows that if $\Omega_0$ is sufficiently small, there are finitely many points $y \in \rho^{-1}(0)$ and a cover $\{\mathcal{M}_{y}\}$ of $\mathcal{M}=\rho^{-1}(\Omega_0)$ such that 
$$\zeta(z) = \sum_y \int_{\mathcal{M}_y}f\circ \rho(\mu)^{-z/2}|\varphi\circ \rho(\mu)| |\rho'(\mu)| \sigma_y(\mu) \,d\mu$$
where $\{\sigma_y\}$ is a partition of unity subordinate to $\{\mathcal{M}_y\}$. Furthermore, the $f_i \circ \rho$ generate the pullback $\rho^*I$ and $f\circ \rho = (f_1\circ \rho)^2 +\cdots +(f_r\circ \rho)^2$. Therefore,
$$
\RLCT_{\mathcal{M}_y}(f \circ \rho; (\varphi \circ \rho)|\rho'|\sigma_y) = \RLCT_{\mathcal{M}_y}(\rho^* I;  (\varphi \circ \rho)|\rho'|)
$$
and the result follows from the two previously displayed equations. 
\end{proof}

We are now ready to prove Theorem~\ref{thm:PolyRLCT} which was inspired by Watanabe.
\begin{proof}[Proof of Theorem~\ref{thm:PolyRLCT}]

Let $Q(\omega) = \sum_{i=1}^{k} (p_i(\omega)-q_i)^2$. The learning coefficient is the RLCT of the Kullback-Leibler distance $K(\omega)$, so it is enough to show that $\RLCT_{\Omega_x}\,K= \RLCT_{\Omega_x}\,Q$ for each $x \in \mathcal{V}(K) = \mathcal{V}(Q)$. By Corollary \ref{thm:AsympBound}, we only need to show that $K$ and $Q$ are equivalent in a sufficiently small neighborhood of $x$.
Now, the Taylor expansion $-\log t = (1-t) + \frac{1}{2}(1-t)^2+\cdots$ implies there are constants $c_1,c_2 > 0$ such that for all $t$ near $1$,
\begin{eqnarray}\label{eq:ltIneq}
c_1 (t-1)^2 \leq -\log t +t-1 \leq c_2(t-1)^2.
\end{eqnarray}
Choosing a sufficiently small $W_x$ such that $p_i(\omega)/q_i$ is near $1$, we have
$$
c_1 (\frac{p_i(\omega)}{q_i}-1)^2 \leq -\log \frac{p_i(\omega)}{q_i} +\frac{p_i(\omega)}{q_i}-1 \leq c_2(\frac{p_i(\omega)}{q_i}-1)^2
$$
for all $\omega \in W_x$. Multiplying by $q_i$, summing from $i=1$ to $k$ and observing that the $p_i$ and the $q_i$ add up to $1$, we get
$$
c_1\sum_{i=1}^{k} q_i \Big(\frac{p_i(\omega)}{q_i}-1 \Big)^2 \leq K(\omega) \leq c_2 \sum_{i=1}^{k} q_i \Big(\frac{p_i(\omega)}{q_i}-1 \Big)^2.
$$
Again, using the fact that the $q_i$ are non-zero, we have
$$
\frac{c_1}{\max_i q_i}\sum_{i=1}^k \big(p_i(\omega)-q_i \big)^2 \leq K(\omega) \leq \frac{c_2}{\min_i q_i} \sum_{i=1}^k \big (p_i(\omega)-q_i \big)^2
$$
which completes the claim. The more general statement for a real analytic $K(\omega)$ which is bounded by scalar multiples of a sum of squared functions follows from Proposition~\ref{thm:GlobalRLCT}, Corollary \ref{thm:AsympBound} and the definition of $\RLCT_\Omega(I;\varphi)$.
\end{proof}

\section{Newton Polyhedra and Nondegeneracy}
\label{sec:Newton}

Given an analytic function $f \in \mathcal{A}_0(\R^d)$, we pick local coordinates $\{w_1, \ldots, w_d\}$ in a neighborhood of the origin. This allows us to represent $f$ as a power series $\sum_{\alpha} c_\alpha \omega^\alpha$ where $\omega = (\omega_1, \ldots, \omega_d)$ and each $\alpha = (\alpha_1,\ldots,\alpha_d)\in \N^d$.  Let $[\omega^\alpha]f$ denote the coefficient $c_\alpha$ of $\omega^\alpha$ in this expansion. Define its \emph{Newton polyhedron} $\mathcal{P}(f) \subset \R^d$ to be the convex hull 
$$
\mathcal{P}(f) ={\rm conv}\,\{\alpha + \alpha': [\omega^\alpha]f \neq 0, \alpha' \in \R^d_{\geq 0} \}.
$$
A subset $\gamma \subset \mathcal{P}(f)$ is a \emph{face} if there exists $\beta \in \R^d$ such that
$$
\gamma  = \{\alpha \in \mathcal{P}(f): \langle \alpha , \beta \rangle \leq \langle \alpha' , \beta \rangle \mbox{ for all } \alpha' \in \mathcal{P}(f)\}. 
$$
where $\langle \,\,, \,\rangle$ is the standard dot product. Dually, the \emph{normal cone} at $\gamma$ is the set of all $\beta \in \R^d$ satisfying the above condition. Each $\beta$ lies in the non-negative orthant $\R^d_{\geq 0}$ because otherwise, the linear function $\langle \,\cdot\,, \beta\rangle$ does not have a minimum over the unbounded set $\mathcal{P}(f)$. As a result, the union of all the normal cones gives a partition $\mathcal{F}(f)$ of the non-negative orthant called the \emph{normal fan}. Now, given a compact subset $\gamma \subset \R^d$, define the \emph{face polynomial}
$$
f_\gamma = \sum_{\alpha \in \gamma \cap \N^d} c_\alpha \omega^\alpha.
$$
Recall that $f_\gamma$ is singular at a point $x \in \R^d$ if $\ord_x f_\gamma \geq 2$, i.e.
$$
f_\gamma(x)= \frac{\partial f_\gamma}{\partial \omega_1}(x) = \cdots =\frac{\partial f_\gamma}{\partial \omega_d}(x) =0.
$$
We say that $f$ is \emph{nondegenerate} if $f_\gamma$ is non-singular at all points in the torus $(\R^*)^d$ for all compact faces $\gamma$ of $\mathcal{P}(f)$, otherwise we say $f$ is \emph{degenerate}. Now, we define the \emph{distance} $l$ of $\mathcal{P}(f)$ to be the smallest $t \geq 0$ such that  $(t, t, \ldots, t) \in \mathcal{P}(f)$. Let the \emph{multiplicity} $\theta$ of $l$ be the codimension of the lowest-dimensional face of $\mathcal{P}(f)$ at this intersection of the diagonal with $\mathcal{P}(f)$. However, if $l=0$, we leave $\theta$ undefined. These
notions of nondegeneracy, distance and multiplicity were first coined
and studied by Varchenko \cite{V}.

We now extend the above notions to ideals. For any ideal $I \subset \mathcal{A}_0$, define
$$
\mathcal{P}(I) = {\rm conv}\,\{\alpha \in \R^d: [\omega^\alpha]f \neq 0
\mbox{ for some }f \in I\}.
$$
Related to this geometric construction is the monomial ideal
$$
\mon(I) = \langle \omega^{\alpha} : [\omega^\alpha]f \neq 0
\mbox{ for some }f \in I \rangle.
$$
Note that $I$ and $\mon(I)$ have the same Newton polyhedron, and if $I$ is generated by $f_1, \ldots, f_r$, then $\mon(I)$ is generated by monomials $\omega^\alpha$ appearing in the $f_i$. One consequence is that $\mathcal{P}(f_1^2+\cdots+f_r^2)$ is the scaled polyhedron $2\mathcal{P}(I)$. More importantly, the threshold of $I$ is bounded by that of $\mon(I)$. To prove this result, we need the
following lemma. Recall that by the
Hilbert Basis Theorem or by Dickson's Lemma \cite{E}, $\mon(I)$ is finitely generated.
\begin{lem}
\label{thm:BoundByGeneratingMonomials}
Given $f \in \mathcal{A}_0(\R^d)$, let $S$ be a finite set of
exponents $\alpha$ of monomials $\omega^\alpha$ which generate
$\mon(\langle f\rangle)$. Then, there is a constant $c>0$ such that $$|f(\omega)| \leq c \sum_{\alpha \in S} |\omega|^\alpha$$ in a sufficiently small neighborhood of the origin.
\end{lem}
\begin{proof}
Let $\sum_\alpha c_{\alpha} \omega^{\alpha}$ be the power series
expansion of $f$. Because $f$ is analytic at the origin, there exists
$\varepsilon > 0$ such that 
$$
\sum_\alpha |c_\alpha|\, \varepsilon^{\alpha_1+\cdots+\alpha_d} < \infty.
$$
Let $S = \{\alpha^{(1)},\ldots, \alpha^{(s)}\}$ where the monomials $\omega^{\alpha^{(i)}}$ generate $\mon(\langle f\rangle)$. Then, 
$$
f(\omega) = \omega^{\alpha^{(1)}} g_1(\omega) + \cdots + \omega^{\alpha^{(s)}} g_s(\omega)
$$
for some power series $g_i(\omega)$. Each series $g_i(\omega)$ is absolutely
convergent in the $\varepsilon$-neighborhood $U$ of the origin because $f$ is
absolutely convergent in $U$. Thus, the $g_i(\omega)$ are 
analytic. Their absolute values are bounded above by
some constant $c$ in $U$, and the lemma follows. \end{proof}

Below, $\RLCT_0(I; \varphi)$ denotes the RLCT of $I$ with respect to $\varphi$ at the origin.

\begin{prop} \label{thm:BoundByMonomial}
Let $I \subset \mathcal{A}_0$  be a finitely generated ideal and $\varphi:\R^d \rightarrow \R$ be nearly analytic at the origin. Then,
$$
\RLCT_0(I; \varphi) \leq \RLCT_0(\mon(I); \varphi).
$$
\end{prop} 
\begin{proof}
Given $f \in \mathcal{A}_0(\R^d)$, let $S$ be a finite set of
generating exponents $\alpha$ for $\mon(\langle f\rangle)$. By Lemma
\ref{thm:BoundByGeneratingMonomials} and the Cauchy-Schwarz inequality,
there exist constants $c, c' > 0$ such that
$$f^2 \leq \Big( c \sum_{\alpha \in S} |\omega|^\alpha \Big)^2  \leq c' \sum_{\alpha \in S} \omega^{2\alpha}$$ 
in a sufficiently small neighborhood of the origin. Therefore, if $f_1, \ldots, f_r$ generate $I$, then $f_1^2+\ldots+f_r^2$ is bounded by a constant multiple of the sum of squares of monomials generating $\mon(I)$. The result now follows from Propostion \ref{thm:AsympUpBound}.
\end{proof}

Given a compact subset $\gamma \subset \R^d$, define the \emph{face ideal}
$$
I_\gamma  = \langle f_\gamma : f \in I \rangle.
$$ 
The next result tells us how to compute $I_\gamma$ 
for an ideal $I = \langle f_1, \ldots, f_r \rangle$.
\begin{prop}
\label{thm4:GeneratorsFaceIdeal}
For all compact faces $\gamma \in \mathcal{P}(I)$, $I_\gamma = \langle f_{1\gamma}, \ldots,
f_{r\gamma}\rangle$.
\end{prop}
\begin{proof}
By definition, $\langle f_{1\gamma}, \ldots, f_{r\gamma} \rangle
\subset I_\gamma$. For the other inclusion, it is enough to show that
$f_\gamma \in \langle f_{1\gamma}, \ldots, f_{r\gamma} \rangle$ for all $f \in I$. First, we claim that if
$\omega^\alpha = \omega^{\alpha'} \omega^{\alpha''}$ with
$\alpha \in
\gamma$ and $\omega^{\alpha'} \in \mon(I)$, then $\omega^{\alpha''} =
1$. Indeed, for all $\beta \in \R^d_{\geq 0}$ in the normal cone dual to
$\gamma$, we have $\langle \alpha,\beta\rangle = \langle
\alpha',\beta\rangle+\langle \alpha'',\beta\rangle$, but $\langle \alpha,\beta\rangle \leq \langle
\alpha',\beta\rangle$ so $\langle \alpha'',\beta\rangle = 0$. This implies that $\alpha' +
k\alpha'' \in \gamma$ for all integers $k > 0$. Since $\gamma$ is
compact, $\alpha''$ must be the zero vector so $\omega^{\alpha''} =
1$.

Now, if $f \in I$, then 
$f = h_1 f_1 + \cdots + h_r f_r$ for some analytic functions $h_1,
\ldots, h_r$. Clearly, $f_\gamma = (h_1 f_1)_\gamma + \cdots + (h_r
f_r)_\gamma$. By the above claim, $(h_i f_i)_\gamma = h_{i0}
f_{i\gamma}$ where $h_{i0}$ is the constant term in $h_i$. Hence, $f_\gamma = h_{10} f_{1\gamma} + \cdots + h_{r0}
f_{r\gamma} \in \langle f_{1\gamma}, \ldots, f_{r\gamma} \rangle$ as required. 
\end{proof}
\begin{remark} Let $\beta$ be a vector in the normal cone dual to the face
$\gamma$ of $\mathcal{P}(I)$. Now, consider the weight order
associated to $\beta$, and let ${\rm in}_\beta f$ be the sum of all the
terms of $f$ that are maximal with respect to this (partial) order \cite[\S
15]{E}. Let ${\rm in}_\beta I$ be the initial ideal
$   {\rm in}_\beta I = \langle {\,\rm in}_\beta f : f \in I\, \rangle$. A set of functions $f_1, \ldots, f_r \in I$ is a \emph{Gr\"{o}bner basis} for $I$ if and only if
$$
{\rm in}_\beta I = \langle {\rm in}_\beta f_1, \ldots, {\rm in}_\beta f_r \rangle.
$$
Comparing this statement with the previous result, one could ask why the generators $f_1, \ldots, f_r$ of $I$ need not be a  Gr\"{o}bner basis for Proposition \ref{thm4:GeneratorsFaceIdeal} to hold. This confusion comes from incorrectly equating the face ideal $I_\gamma$ with the
initial ideal  ${\rm in}_\beta I$ when we only have containment $I_\gamma \subset {\rm in}_\beta I$. For instance, if 
$$
I = \langle f_1, f_2, f_3 \rangle = \langle  xy-z^3, xz-y^3, yz-x^3 \rangle
$$
and $\gamma$ is the convex hull of $\{(1,1,0), (1,0,1), (0,1,1)\}$, then $I_\gamma = \langle xy,xz,yz \rangle$. Meanwhile, $\beta = (1,1,1)$ and ${\rm in}_\beta I$ contains $y^4-z^4 = zf_1-yf_2$ but $y^4-z^4 \notin I_\gamma$.
\end{remark}

Lastly, we give several equivalent definitions of
nondegeneracy for ideals. If an ideal $I$ satisfies these conditions, then we say that $I$ is \emph{sos-nondegenerate}, where \emph{sos}
stands for \emph{sum-of-squares}. 

\begin{prop}
\label{thm:EquivalentDefinitionIdealRLCT}
Let $I \subset \mathcal{A}_0$ be an ideal. The following are equivalent:
\begin{enumerate}
\item For some
generating set $\{f_1, \ldots, f_r\}$ for $I$, $f_1^2+\cdots+f_r^2$
is nondegenerate.
\item For all
generating sets $\{f_1, \ldots, f_r\}$ for $I$, $f_1^2+\cdots+f_r^2$
is nondegenerate.
\item For all compact faces $\gamma \subset \mathcal{P}(I)$, the variety
$\mathcal{V}(I_\gamma) \subset \R^d$ does not intersect the torus $(\R^*)^d$.
\end{enumerate}

\end{prop}
\begin{proof}
Let $f_1, \ldots, f_r$ generate $I$ and let $f =
f_1^2+\cdots+f_r^2$. If $\gamma$ is a compact face of
$\mathcal{P}(I)$, then the set $(2\gamma)$ is a compact face of
$\mathcal{P}(f)=2\mathcal{P}(I)$. Furthermore, $f_{(2\gamma)} =
f_{1\gamma}^2+\cdots+f_{r\gamma}^2$ and
\begin{align*}
  \frac{\partial f_{(2\gamma)}}{\partial \omega_i} &= 2f_{1\gamma}
  \frac{\partial f_{1\gamma}}{\partial \omega_i} + \cdots +
  2f_{r\gamma} \frac{\partial f_{r\gamma}}{\partial \omega_i}.
\end{align*}
Now, $f_{1\gamma}^2+\cdots+f_{r\gamma}^2=0$ if and only if
$f_{1\gamma}=\cdots=f_{r\gamma} =0$. It follows that $f$ is nondegenerate if and only if $\mathcal{V}(
\langle f_{1\gamma}, \ldots, f_{r\gamma}\rangle) \cap (\R^*)^d =\mathcal{V}(
I_\gamma) \cap (\R^*)^d = \emptyset$ for all compact faces $\gamma
\subset \mathcal{P}(I)$. This proves $(1)
\Leftrightarrow (3)$ and $(2)
\Leftrightarrow (3)$.
\end{proof}

\begin{remark}
The nondegeneracy of a function $f$ need not imply the sos-non-degeneracy of the ideal $\langle f \rangle$, e.g. $f = x+y$. 
\end{remark}

\begin{remark}
After finishing this paper, we discovered another notion of
nondegeneracy for ideals of \emph{complex} formal power series due to Saia \cite{Saia}, which was shown to be equivalent to the complex version of Proposition \ref{thm:EquivalentDefinitionIdealRLCT}(3) \cite[\S 2]{BiviaAusina}.
\end{remark}

We recall some basic facts about toric varieties. We say a polyhedral
cone $\sigma$ is generated by vectors $v_1, \ldots, v_k \in \R^d$ if
$\sigma = \{\sum_i \lambda_i v_i : \lambda_i \geq 0 \}$. If $\sigma$
is generated by lattice vectors $v_i \in \Z^d$, then $\sigma$ is
\emph{rational}.  If the origin is a face of $\sigma$, then $\sigma$
is \emph{pointed}. A \emph{ray} is a pointed one-dimensional
cone. Every rational ray has a lattice generator of minimal length
called the \emph{minimal generator}, and every pointed rational
polyhedral cone $\sigma$ is generated by the minimal generators of its one-dimensional faces. If these minimal generators are linearly independent
over $\R$, then $\sigma$ is \emph{simplicial}. A simplicial cone is \emph{smooth} if its minimal generators also form part of a $\Z$-basis of
$\Z^d$. A collection $\mathcal{F}$ of pointed rational polyhedral
cones in $\R^d$ is a \emph{fan} if the faces of every cone in
$\mathcal{F}$ are in $\mathcal{F}$ and the intersection of any two
cones in $\mathcal{F}$ are again in $\mathcal{F}$. The \emph{support}
of $\mathcal{F}$ is the union of its cones as subsets of $\R^d$. If
the support of $\mathcal{F}$ is the non-negative orthant, then
$\mathcal{F}$ is \emph{locally complete}.  If every cone of
$\mathcal{F}$ is simplicial (resp. smooth), then $\mathcal{F}$ is
\emph{simplicial} (resp. \emph{smooth}). A fan $\mathcal{F}_1$ is a
\emph{refinement} of another fan $\mathcal{F}_2$ if the cones of
$\mathcal{F}_1$ come from partitioning the cones of
$\mathcal{F}_2$. See \cite{F} for more details. 

Given a smooth simplicial locally complete fan $\mathcal{F}$, we have a smooth toric variety $\PP(\mathcal{F})$ covered by open charts $U_\sigma \simeq \R^d$, one for each cone $\sigma$ of $\mathcal{F}$ that is maximal under inclusion. Furthermore, the \emph{blow-up} $\rho_{\mathcal{F}}: \PP(\mathcal{F}) \rightarrow \R^d$ is defined as follows: for each maximal cone $\sigma$ of $\mathcal{F}$ minimally generated by $v_1, \ldots, v_d$ with $v_i = (v_{i1}, \ldots, v_{id})$, we have monomial maps $\rho_\sigma: U_\sigma \rightarrow \R^d$ on the open charts.
\begin{eqnarray*}
&(\mu_1, \ldots, \mu_d) \mapsto (\omega_1, \ldots, \omega_d)& \\
&\omega_1 = \mu_1^{v_{11}} \mu_2^{v_{21}} \cdots \mu_d^{v_{d1}}& \\
&\omega_2 = \mu_1^{v_{12}} \mu_2^{v_{22}} \cdots \mu_d^{v_{d2}}& \\
&\vdots& \\
&\omega_d = \mu_1^{v_{1d}} \mu_2^{v_{2d}} \cdots \mu_d^{v_{dd}}& 
\end{eqnarray*}
Let $v=v_\sigma$ be the matrix $(v_{ij})$  where each minimal generator $v_i$ forms a row of $v$. We represent the above monomial map by $\omega = \mu^{v}$. If $v_{i+}$ represents the $i$-th row sum of $v$, the Jacobian determinant of this map is 
$$(\det v) \mu_1^{v_{1+}-1} \cdots \mu_d^{v_{d+}-1}.$$

We are now ready to connect these concepts. The next two theorems are~due to Varchenko, see \cite{V} and \cite[\S 8.3]{AGV}. His notion of degeneracy is weaker than~ours because it does not include the condition $f_\gamma = 0$, but his proof \cite[Lemma 8.9]{AGV} actually supports the stronger notion. The set up is as follows: suppose $f$ is analytic in a neighborhood $W$ of the origin. Let $\mathcal{F}$ be any smooth simplicial refinement of the normal fan $\mathcal{F}(f)$ and $\rho_{\mathcal{F}} $ be the blow-up associated to $\mathcal{F}$. Set $M =\rho_{\mathcal{F}}^{-1}(W)$. Let $l$ be the distance of $\mathcal{P}(f)$ and $\theta$ its multiplicity. 

\begin{thm}
\label{thm:VarchenkoDesingular}
$(M,W,\rho_{\mathcal{F}})$ desingularizes $f$ at $0$ if $f$ is nondegenerate. 
\end{thm}

\begin{thm}
\label{thm:VarchenkoDistance}
$\RLCT_0\,f = (1/l, \theta)$ if $(M,W,\rho_{\mathcal{F}})$ desingularizes $f$ at $0$.
\end{thm}

We extend Theorem \ref{thm:VarchenkoDistance} to compute $\RLCT_0(f;\omega^\tau)$ for monomials $\omega^\tau$. Given a polyhedron $\mathcal{P}(f) \subset \R^d$ and a vector $\tau =(\tau_1, \ldots, \tau_d)$ of non-negative integers, let the \emph{$\tau$-distance} $l_\tau$ be the smallest $t \geq 0$ such that $t(\tau_1+1,\ldots, \tau_d+1) \in \mathcal{P}(f)$ and let the multiplicity $\theta_\tau$ be the codimension of the face at this intersection.
\begin{thm}
\label{thm:ScaledDistance}
$\RLCT_0(f;\omega^\tau) = (1/l_\tau, \theta_\tau)$ if $(M,W,\rho_{\mathcal{F}})$ desingularizes $f$ at~$0$.
\end{thm}
\begin{proof}
We follow roughly the proof in \cite[\S 8]{AGV} of Theorem \ref{thm:VarchenkoDistance}. Let $\sigma$ be a maximal cone of $\mathcal{F}$. Because $\mathcal{F}$ refines $\mathcal{F}(f)$, $\sigma$ is a subset of some maximal cone $\sigma'$ of $\mathcal{F}(f)$. Let $\alpha \in \R^d$ be the vertex of $\mathcal{P}(f)$ dual to $\sigma'$. Let $v$ be the matrix whose rows are minimal generators of $\sigma$ and $\rho$ the monomial map $\mu \mapsto \mu^v$.  Under this map, the term $c_\alpha w^\alpha$ in $f$ becomes the leading monomial, so $f(\rho(\mu)) = g(\mu) \mu^{v\alpha}$ for some function $g$ satisfying $g(\mu)\neq 0$ for all $\mu \in U_\sigma$.
Then,  
\begin{eqnarray*}
  |f(\omega)|^{-z}|\omega^\tau| \,d\omega &=& |f(\rho(\nu))|^{-z}|\rho(\mu)|^\tau
  |\rho'(\mu)|\,d\mu \\
&=& (\det v) |g(\mu)|^{-z} |\mu|^{-v\alpha z} |\mu^{v\tau} \mu_1^{v_{1+}-1} \cdots \mu_d^{v_{d+}-1}|
\end{eqnarray*}
Thus, for the cone $\sigma$,
$$
(\lambda_\sigma, \theta_\sigma) = (\min S, \nmin S), \quad S = \Big\{\frac{ v_i \cdot (\tau+1) }{ v_i \cdot \alpha}: 1\leq i \leq d\Big\}
$$
where $\tau+1 = (\tau_1+1, \ldots, \tau_d+1)$. We now give an interpretation for the elements of $S$. Fixing $i$, let $P$ be the affine hyperplane normal to $v_i$ passing through $\alpha$. Then, $ (v_i \cdot \alpha)/( v_i \cdot (\tau+1))$ is the distance of $P$ from the origin along the ray $\{t(\tau+1) : t\geq 0\}$. Since $\RLCT_0(f;\omega^\tau) = \min_\sigma (\lambda_\sigma, \theta_\sigma)$, the result~follows.
\end{proof}
\begin{remark}
After finishing this paper, the author discovered that a similar result
was proved by Vasil'ev \cite{Vasilev} for
\emph{complex} analytic functions. \qed
\end{remark}

Monomial ideals play in special role in the theory of real log
canonical thresholds of ideals. The proof of this next result is due to Piotr Zwiernik. 
\begin{prop}
\label{thm:MonomialRLCT}
Monomial ideals are sos-nondegenerate.
\end{prop}
\begin{proof}
Let $f = f_1^2+\cdots+f_r^2$ where $f_1, \ldots, f_r$ are monomials
generating $I$. For each face $\gamma$ of $\mathcal{P}(I)$, $f_\gamma$
is also a sum of squares of monomials, so $f_\gamma$ does not have any
zeros in $(\R^*)^d$ and the result now follows from Proposition \ref{thm:EquivalentDefinitionIdealRLCT}(3).
\end{proof}

Our tools now allow us to prove Theorem \ref{thm:SosnondegenerateRLCT}. As a special case, we~have a formula for the RLCT of a monomial ideal with respect to a monomial amplitude function. The analogous formula for \emph{complex} log canonical thresholds~of monomial ideals was discovered and proved by Howald \cite{H}.


\begin{proof}[Proof of Theorem~\ref{thm:SosnondegenerateRLCT}]
If the ideal $I$ is sos-nondegenerate, then the equality follows from Proposition \ref{thm:EquivalentDefinitionIdealRLCT}, Theorem \ref{thm:VarchenkoDesingular} and Theorem \ref{thm:ScaledDistance}. For all other ideals, the inequality is the result of Proposition \ref{thm:BoundByMonomial} and Proposition \ref{thm:MonomialRLCT}.
\end{proof}

\begin{remark}
Define the principal part $f_\mathcal{P}$ of $f$ to be $\sum_\alpha c_\alpha \omega^\alpha$ where the sum is over all $\alpha$ lying in some compact face $\gamma$ of $\mathcal{P}(f)$. The above theorems imply that if $f$ is nondegenerate, then $\RLCT_0\,f = \RLCT_0\,f_\mathcal{P}$. However, the latter is not true in general. For instance, if $f = (x+y)^2+y^4$, then $f_\mathcal{P} = (x+y)^2$ but $\RLCT_0\,f = (3/4,1)$ and $\RLCT_0\,f_\mathcal{P} = (1/2,1)$.
\end{remark}

Our first corollary shows that the BIC is a special case of Theorem \ref{thm:Watanabe}.
\begin{cor}
If $f \in \mathcal{A}_0(\R^d)$ has a local minimum at the origin with $f(0)=0$ and its Hessian $(\partial^2 f/\partial \omega_i \partial \omega_j)$ is full rank, then $\RLCT_0\, f = (d/2,1)$.
\end{cor}
\begin{proof}
Because its Hessian is full rank, there is a linear change of variables such that $f = \omega_1^2+\cdots+\omega_d^2+O(\omega^3)$. Thus, $f$ is nondegenerate and the Newton polyhedron $\mathcal{P}(f)$ has distance $l=2/d$ with $\theta=1$.
\end{proof}
\begin{cor}
Let $I$ be the ideal $\langle f_1, \ldots, f_s \rangle$, and suppose the Jacobian matrix $(\partial f_i/\partial \omega_j)$ has rank $r$ at $0$. Then, $\RLCT_0\,I \leq (\frac{1}{2}(r+d), 1)$.
\end{cor}
\begin{proof}
Because the rank of $(\partial f_i/\partial \omega_j)$ is $r$, there is a linear change of variables such that the only linear monomials appearing in $I$ are $\omega_1, \ldots, \omega_r$. It follows that $\mathcal{P}(I)$ lies in the halfspace 
$\textstyle \alpha_1+\cdots + \alpha_r + \frac{1}{2}(\alpha_{r+1}+\cdots+\alpha_d) \geq 1$ and its distance is at least $1/(r+\frac{d-r}{2}) =  2/(r+d)$. 
\end{proof}

\section{Applications to Statistical Models}
\label{sec:Examples}

In this section, we use our tools to compute the learning coefficients of
a na\"{i}ve Bayesian network $\mathcal{M}$ with two
ternary random variables and two hidden states. It was designed by Evans,
Gilula and Guttman \cite{EGG} for~investigating connections between the
recovery time of 132 schizophrenic patients and the
frequency of visits by their relatives. Their data is summarized in the
$3{\times}3$ contingency table 
$$ \begin{matrix}
& 2 {\leq} Y {<} 10 & 10 {\leq} Y {<} 20 & 20 {\leq} Y & &{\it Totals} \cr
\hbox{Visited regularly \!\!} & 43 & 16 & 3 && {\it 62} \cr
\hbox{Visited rarely}      & 6 & 11 & 10  && {\it 27} \cr
\hbox{Visited never}     & 9 &  18 & 16 && {\it 43} \cr
{\it Totals}& {\it 58} & {\it 45} & {\it 29} & & {\bf 132} \cr
\end{matrix}
$$
which we store as a $3{\times}3$ matrix $\hat{q}$ of relative frequencies. The model is given by
\begin{align*}
p: \quad \Omega &=
\Delta_1{\times}\Delta_2{\times}\Delta_2{\times}\Delta_2{\times}\Delta_2 \quad
\rightarrow\quad \Delta_8\\
 \omega &= (t,a_1, a_2, b_1, b_2, c_1, c_2, d_1, d_2) \quad \mapsto\quad (p_{ij})\\
p_{ij} &= ta_ib_j + (1-t)c_id_j,\quad i,j \in \{0,1,2\} 
\end{align*}
where $a_0=1\!-\!a_1\!-\!a_2$, $a=(a_0, a_1, a_2)\! \in\! \Delta_2$
and similarly for $b, c$ and $d$. Hence, a $3{\times}3$ matrix in
the model is a convex combination of two rank one matrices, so it has
rank at most two. The marginal likelihood of the data (after discarding a constant multiplicative factor) is the integral
\begin{align*}
 \mathcal{I}= \int_{\Omega} p_{00}^{43}\, p_{01}^{16}\, p_{02}^{3} \,p_{10}^{6}\,
  p_{11}^{11} \,p_{12}^{10} \,p_{20}^{9} \,p_{21}^{18} 
\,p_{22}^{16}\,\, d\omega
\end{align*}
which was computed exactly by Sturmfels, Xu and the author \cite{LSX}.

We now estimate
this integral
using Watanabe's asymptotic formula for the log likelihood integral in Theorem
\ref{thm:Watanabe}. We assume that the data $\hat{q}$ was generated by some true distribution $q = (q_{ij}) \in \R^{3{\times}3}$ in the model. 
Ideally, we want $q$ to be equal to the matrix $\hat{q}$ of 
relative frequencies, but in general, the data $\hat{q}$ rarely lies
in the model. In this example, the matrix
$\hat{q}$ is not in the model because it is full rank. However, we \emph{should} be able to find a distribution $q$ in the model
that is close to $\hat{q}$, because in practice, we want to study models which
describe the data well. A good candidate for $q$ is the maximum
likelihood distribution. Using the EM algorithm, this distribution is
\begin{align*}
q = \frac{1}{132}\left(\begin{array}{ccc}
43.00153927 & 15.99813189 & 3.000328847 \\
5.979732739  &  11.12298188 &   9.897285383 \\
9.018728012  &  17.87888620   & 16.10238577
\end{array}\right)
\end{align*}
which comes from the maximum likelihood estimate
\begin{align*}
t &= 0.5129202328\\
                             (a_1,a_2) &= (0.09139459898, 0.3457903589),\\
                             (b_1,b_2) &= (0.1397061214, 0.4386217768), \\
                              (c_1,c_2) &= (0.8680689680,
                              0.05580725171), \\
                              (d_1,d_2) &= (0.7549807403,0.2380125694).
\end{align*}
Note that the ML distribution $q$ is indeed very close to the data $\hat{q}$. 

Our next theorem summarizes how the asymptotics of $\log Z(N)$ depend on
$q$. Let $S_i$ denote the set of rank $i$
matrices in $p(\Omega)$, and $S_i^* \subset S_i$ be the matrices with positive entries. Before we prove this theorem, let us apply it to our
statistical problem. Using the exact value of $\mathcal{I}$ computed by Lin--Sturmfels--Xu \cite{LSX}, we~get
\begin{align*}
  (\,\log \mathcal{I}\,)_{\rm exact} \,\,=\,\, -273.1911759.
\end{align*} 
Meanwhile, if the BIC was erroneously applied with the dimension $d=9$ of the parameter space, we
would have
\begin{align*}
  (\,\log \mathcal{I}\,)_{\rm BIC} \,\, = \,\, -280.7992160.
\end{align*} 
On the other hand, by calculating the real log canonical threshold of the polynomial ideal $\langle p(\omega) - q\rangle$,  we find that the learning coefficient of the model at the ML distribution $q$ is $(\lambda, \theta) = (7/2,1)$. This gives us the approximation
\begin{align*}
  (\,\log \mathcal{I}\,)_{\rm RLCT}\,\, \approx \,\, -275.9164140
\end{align*}
which is closer than
the BIC to
the exact value of $\log \mathcal{I}$.

Our proposal to use the ML distribution as the true distribution $q$ is admittedly simplistic, given that noise in the data will almost surely bring us to some $q \in S_2^*$. Nonetheless, our next theorem proves that the learning coefficient is always smaller than the $(9/2, 1)$ prescribed by the BIC. For deeper statistical discussions, the reader should turn to Drton and Plummer \cite{DP} where they addressed the paradox of circular reasoning in requiring true parameter values for the asymptotic approximation of Bayesian integrals. They also proposed a novel algorithm where the marginal likelihood is estimated as a weighted average of the contributions from all true distributions. We hope that mathematical analyses such as our next theorem will help inform these kinds of discussions, and provide useful estimates and bounds for a variety of statistical computations.

\begin{thm} \label{thm:SchizoCoef} 
The learning coefficient $(\lambda, \theta)$ of the model at $q>0$ is given by
$$
  (\lambda, \theta) = \left\{
\begin{array}{ll}
(5/2,1) & \mbox{ if }q \in S_1^*,  \\
(7/2,1) & \mbox{ if }q \in S_{2}^*.
\end{array}
\right.
$$
Therefore, asymptotically as $N \rightarrow \infty$, 
$$\log Z(N) \,\,=\,\, N \sum_{i,j} \hat{q}_{ij} \log q_{ij} -
\lambda \log N + (\theta - 1) \log \log N \,\,+ \,\,\eta_N$$
where $\hat{q}$ is the matrix of relative frequencies of the data and $\eta_N$ is a random variable whose expectation $\mathbb{E}[\eta_N]$ converges to a constant.
\end{thm}

We postpone the proof of this theorem to the end of the
section. Let us begin with a few remarks about our approach to
this problem. Firstly, Theorem
\ref{thm:PolyRLCT} states that the learning coefficient $(\lambda, \theta)$ of the statistical model is given by
$$(2\lambda, \theta) = \min_{\omega^* \in \mathcal{V}}
\RLCT_{\Omega_{\omega^*}} \,\langle p(\omega) - q \rangle $$
where $\mathcal{V}$ is the fiber $ p^{-1}(q) = \{\omega \in \Omega : p(\omega) =
q\}$ over $q$. Instead of focusing on a
fixed $q$ and its fiber $\mathcal{V}$, let us vary the parameter $\omega^*$ over all
of $\Omega$. For each
$\omega^* \in \Omega$, we translate $\Omega$ so that $\omega^*$
is the origin and compute the RLCT of the ideal $\langle
p(\omega+\omega^*) - p(\omega^*)\rangle$. This is the content of
Proposition \ref{thm:SchizoProp}. The proof of Theorem
\ref{thm:SchizoCoef} will then consist of minimizing these RLCTs over
the fiber $\mathcal{V}$ for each $q$ in the model.

Secondly, in our computations, we will often be choosing different
generators for our ideal and making appropriate changes of
variables. Generators with few terms
and small total degree are often highly desired. Another useful trick is to multiply or divide the
generators by functions $f(\omega)$ satisfying $f(0) \neq 0$. Such
functions are units in the ring $\mathcal{A}_0$ of real analytic
functions so this multiplication or division will not change the ideal
generated. 

We will perform many of the computations by hand to demonstrate how~the various properties from Section \ref{sec:RLCT} can be applied. At points in the proof where RLCTs of monomial ideals are required, the {\tt Singular} library from Section~\ref{sec:Introduction} comes in useful. We hope that some day the computation of learning coefficients for statistical models will be fully automated.

Thirdly, for the full proof of Proposition~\ref{thm:SchizoProp}, we will have to analyze interactions between the model singularities and the boundary of the parameter space. Some of these interactions are messy. To improve the readability of the paper, we moved the detailed proof to the appendix while retaining some interesting computations in this section.

Finally, we come to our main proposition. Let us define the following subsets of~$\Omega$. These subsets stratify $\Omega$ according to the real log canonical threshold in the manner described in Conjecture \ref{conj:stratification}.
$$
\begin{array}{lcl}
\Omega_{u} &=& \{\omega^* \in \Omega : t^* \in \{0,1\} \}\\
\Omega_{m} &=& \{\omega^* \in \Omega : t^* \notin \{0,1\} \}\\
\Omega_{m0} &=& \{\omega^* \in \Omega_m : a^* = c^*, b^* = d^* \} \\
\Omega_{m0kl} &=& \{\omega^* \in \Omega_{m0} : \#\{i:a_i^* = 0\}=k, \#\{i:b_i^* = 0\}=l\} \\
\Omega_{m1} &=& \{\omega^* \in \Omega_m : (b^*\neq d^* , a^* = c^*) \mbox{ or } (a^* \neq c^*, b^* = d^*)\} \\
\Omega_{m10} &=& \{\omega^* \in \Omega_{m1} : (a^* = c^*,\exists\, i\,\, a_i^*=0) \mbox{ or }( b^* = d^*, \exists\, i\,\, b_i^* = 0) \} \\
\Omega_{m2} &=& \{\omega^* \in \Omega_m : a^*\neq c^*, b^* \neq d^* \} \\
\Omega_{m2ad} &=& \{\omega^* \in \Omega_{m2} : \exists\, i,j\,\, a_i^*=d_j^*=0, c_i^* \neq 0, b_j^* \neq 0 \} \\
\Omega_{m2bc} &=& \{\omega^* \in \Omega_{m2} : \exists\, i,j\,\, b_i^*=c_j^*=0, d_i^* \neq 0, a_j^* \neq 0 \} \\
\Omega_{m21} &=& \Omega_{m2ad} \cup \Omega_{m2bc} \\
\Omega_{m22} &=& \Omega_{m2ad} \cap \Omega_{m2bc}.
\end{array}
$$

\begin{prop} \label{thm:SchizoProp}
Given $\omega^* \in \Omega$, let $I$ be the ideal $\langle p(\omega+\omega^*) - p(\omega^*) \rangle$. Then,
$$
\RLCT_0\, I = \left\{
\begin{array}{ll}
(5,1) & \mbox{ if }\omega^* \in \Omega_{u},  \\
(6,2) & \mbox{ if } \omega^* \in \Omega_{m000}, \\
(6,1) & \mbox{ if } \omega^* \in \Omega_{m010} \cup \Omega_{m001} \cup \Omega_{m020} \cup \Omega_{m002}, \\
(7,2) & \mbox{ if } \omega^* \in \Omega_{m011}, \\
(7,1) & \mbox{ if } \omega^* \in \Omega_{m012} \cup \Omega_{m021}, \\
(8,1) & \mbox{ if } \omega^* \in \Omega_{m022} , \\
(6,1) & \mbox{ if }\omega^* \in \Omega_{m1} \setminus \Omega_{m10},\\
(7,1) & \mbox{ if }\omega^* \in \Omega_{m10},\\
(7,1) & \mbox{ if }\omega^* \in \Omega_{m2} \setminus\Omega_{m21}, \\
(8,1) & \mbox{ if }\omega^* \in  \Omega_{m21} \setminus \Omega_{m22}, \\
(9,1) & \mbox{ if }\omega^* \in \Omega_{m22}.
\end{array}
\right.
$$
\end{prop}

\begin{proof}[Proof Idea] We give a shortened analysis that ignores the effect of the boundary of $\Omega$ on the RLCTs. The derived RLCTs will be smaller than the actual~ones by Proposition~\ref{thm:EffectOfBoundary}. A full proof involving boundary effects is given in the appendix.

Our ideal $I$ is generated by $g_{ij} = f_{ij}(\omega+\omega^*) - f_{ij}(\omega^*)$ where
$$
f_{ij} = ta_ib_j + (1-t)c_id_j, \quad i,j \in \{0,1,2\}
$$
and $a_0=b_0=c_0=d_0=1$. One can check that $I$ is also generated by $g_{10}, g_{20}$, $g_{01}, g_{02},$ and 
$
g_{ij}-(d_j+d_j^*)g_{i0}-(a_i+a_i^*)g_{0j}, i,j\in\{1,2\}
$
which expand to give
$$
\begin{array}{c}
c_1(t_1^*-t) + a_1(t_0^*+t)+tu_1^*\\
c_2(t_1^*-t) + a_2(t_0^*+t)+tu_2^*\\
d_1(t_1^*-t) + b_1(t_0^*+t)+tv_1^*\\
d_2(t_1^*-t) + b_2(t_0^*+t)+tv_2^*\\
a_1d_1-a_1t_0^*v_1^*+d_1t_1^*u_1^*\\
a_1d_2-a_1t_0^*v_2^*+d_2t_1^*u_1^*\\
a_2d_1-a_2t_0^*v_1^*+d_1t_1^*u_2^*\\
a_2d_2-a_2t_0^*v_2^*+d_2t_1^*u_2^*
\end{array}
$$
where $t_0^* = t^*, t_1^*=1-t^*, u_i^* = a_i^*-c_i^*, v_i^* = b_i^*-d_i^*$. Note that $\sum (a_i + a_i^*) =  1$ and $\sum a_i^* = 1$ so $\sum a_i = 0$ and similarly for $b, c, d$. Also, $\sum u_i^* = \sum a_i^* - c_i^* = 0$. The same is true for $v^*$. We now do a case-by-case analysis.

$ $\\
\noindent \textbf{Case 1: $\omega^* \in \Omega_m$.}
 
This implies $t_0^* \neq 0$ and $t_1^* \neq 0$. Since the indeterminates $b_1,b_2,c_1,c_2$ appear only in the first four polynomials, this suggests the change of variables
$$
\begin{array}{rl}
c_i =& (c_i' -tu_i^*- a_i(t_0^*+t))/(t_1^*-t), \quad i = 1,2\\
b_i =& (b_i' -tv_i^*- d_i(t_1^*-t))/(t_0^*+t), \quad i = 1,2 \\
\end{array}
$$
with new indeterminates $t,a_1,a_2,b_1',b_2',c_1',c_2',d_1,d_2$. In view of Proposition \ref{thm:ChangeOfVar}, the Jacobian determinant of this substitution is a constant, while the pullback ideal can be written as $I_1+I_2$ where $I_1= \langle b_1', b_2', c_1', c_2'\rangle$ and $I_2$ is generated by
$$
\begin{array}{c}
a_1d_1-a_1t_0^*v_1^*+d_1t_1^*u_1^*,\\
a_1d_2-a_1t_0^*v_2^*+d_2t_1^*u_1^*,\\
a_2d_1-a_2t_0^*v_1^*+d_1t_1^*u_2^*,\\
a_2d_2-a_2t_0^*v_2^*+d_2t_1^*u_2^*.
\end{array}
$$
The indeterminates in $I_1$ and $I_2$ are disjoint, so we may apply Proposition~\ref{thm:DisjointVars}. The RLCT of $I_1$ is $(4,1)$. Now, we focus on computing the RLCT of $I_2$.

$ $\\
\noindent \textbf{Case 1.1: $\omega^* \in \Omega_{m1}$.} 

This implies $u^*\neq 0, v^*=0$ or $u^*= 0, v^*\neq 0$. Without loss of generality, we assume $v^*=0, u_1^*\neq 0, u_2^*\neq 0$ ($u_1^*+u_2^*+u_3^*=0$ so at most one of them is zero) and substitute
$$
d_i = (d_i'+a_1t_0^*v_i^*)/(t_1^*u_1^*+a_1), \quad i =1,2. \\
$$ The resulting pullback of $I_2$ is $\langle d_1',d_2'\rangle$. If $\omega^*$ lies in the interior of $\Omega$, we use~either Newton polyhedra or Proposition \ref{thm:DisjointVars} to show that the RLCT of this monomial ideal is $(2,1)$. If $\omega^*$ lies on the boundary of $\Omega$, the situation is more complicated and we analyze it in detail in the appendix.

$ $\\
\noindent \textbf{Case 1.2: $\omega^* \in \Omega_{m2}$.}

This implies $u^*\neq 0, v^*\neq 0$. Without loss of generality, suppose that $u_1^* \neq 0$. If $\omega^* \in \Omega_{m21}$, we further assume that $a_1^*=d_j^*=0, u_1^* \neq 0, v_j^* \neq 0$. Substituting
$$
\begin{array}{rl}
d_i = & (d_i'+a_1t_0^*v_i^*)/(a_1+t_1^*u_1^*), \quad i =1,2 \\
a_2 = & (a_2' + a_1 u_2^*)/u_1^*,
\end{array}
$$
the pullback ideal is $\langle a_2',d_1',d_2'\rangle$ so the RLCT at an interior point is $(3,1)$. 

$ $\\
\noindent \textbf{Case 1.3: $\omega^* \in \Omega_{m0}$.} 

This implies $u_i^* = v_i^* =0$ for all $i$. The pullback ideal can be written as 
$$\langle a_1, a_2 \rangle \langle d_1, d_2\rangle$$ 
whose RLCT over an interior point of $\Omega$ is $(2,2)$ by Proposition \ref{thm:DisjointVars}. 

$ $\\
\noindent \textbf{Case 2: $\omega^* \in \Omega_{u}$.} 

Without loss of generality, assume $t^*=0$ and substitute
$$
\begin{array}{rll}
c_i &= (c_i'- t(a_i+u_i^*))/(1-t) &i = 1,2\\
d_i &= (d_i'- t(b_i+v_i^*))/(1-t) &i = 1,2.\\
\end{array}
$$
The pullback ideal is the sum of $\langle c_1', c_2', d_1', d_2' \rangle$ and
$$
\langle t \rangle \langle a_1+u_1^*, a_2+u_2^* \rangle \langle b_1+v_1^*,b_2+v_2^*\rangle .
$$
The RLCT of the first summand is $(4,1)$. The RLCT of $\langle
t \rangle$ is $(1,1)$ while~that of $\langle a_1+u_1^*, a_2+u_2^*
\rangle$ and $\langle b_1+v_1^*,b_2+v_2^*\rangle$ are at least $(2,1)$
each. By Proposition~\ref{thm:DisjointVars}, the RLCT of their product
is $(1,1)$ and that of the pullback ideal is $(5,1)$.
\end{proof}

\begin{proof}[Proof of Theorem \ref{thm:SchizoCoef}]
Given a matrix $q = (q_{ij})$, the learning coefficient~$(\lambda,
\theta)$ of the model at $q$ is the minimum of RLCTs at points
$\omega^* \in \Omega$ where $p(\omega^*) = q$. The theorem then follows
from Proposition \ref{thm:SchizoProp}, Theorem \ref{thm:Watanabe} and the claims 
\begin{align*}
  p(\Omega_{u}) = S_1, \,\,\, p(\Omega_{m0}) \subset S_1, \,\,\, 
  p(\Omega_{m1}) \subset S_1, \,\,\,  p(\Omega_{m21}) \notin S_{2}^*.
\end{align*}
These four claims are easy to check from the definitions of the subsets of $\Omega$.
\end{proof}

\section{Appendix}

In this section, we give a full proof of Proposition~\ref{thm:SchizoProp} that considers the effect of the boundary of the parameter space $\Omega$ on the RLCTs. The next lemma comes in handy in dealing with
boundary issues. It helps us in computing the RLCTs of monomial ideals at boundary points where the parameter space contains a nice neighborhood $\Omega_1 \times \Omega_2$. Here, $\Omega_1$ is an orthant in the coordinates involved in the monomials of $I$, while $\Omega_2$ is a small cone in the remaining coordinates.
\begin{lem}\label{thm:EpsCone}
Let $\Omega \subset \{(x_1, \ldots, x_d) \in \R^d\}$ be semianalytic. Let $I$ be a monomial ideal and $\varphi$ a monomial function in $x_1, \ldots, x_r$. If there exists a vector $\xi \in \R^{d-r}$ such that $\Omega_1 {\times}\Omega_2 \subset \Omega$ for sufficiently small $\varepsilon$,
$$
\begin{array}{rl}
\Omega_1 &= \{(x_1, \ldots, x_r) \in [0, \varepsilon]^{r} \} \\
\Omega_2 &= \{(x_{r+1}, \ldots, x_d) = t (\xi+\xi ') \mbox{ for } t \in [0,\varepsilon] , \xi' \in [-\varepsilon, \varepsilon]^{d-r} \},
\end{array}
$$
then $\RLCT_{\Omega_0} (I;\varphi) = \RLCT_0(I; \varphi)$.
\end{lem}
\begin{proof}
Because $I$ and $|\varphi|$ remain unchanged by the flipping of signs of $x_1, \ldots, x_r$, their threshold does not depend on the choice of orthant, so $\RLCT_{\Omega_1}(I; \varphi)$ = $\RLCT_0(I; \varphi)$. The lemma now follows from Proposition \ref{thm:DisjointVars} and the fact that the threshold of the zero ideal over the cone neighborhood $\Omega_2$ is $(\infty, -)$.
\end{proof}

\begin{proof}[Detailed Proof of Proposition~\ref{thm:SchizoProp}]
Recall that the ideal $I$ is generated by 
$$
\begin{array}{c}
c_1(t_1^*-t) + a_1(t_0^*+t)+tu_1^*\\
c_2(t_1^*-t) + a_2(t_0^*+t)+tu_2^*\\
d_1(t_1^*-t) + b_1(t_0^*+t)+tv_1^*\\
d_2(t_1^*-t) + b_2(t_0^*+t)+tv_2^*\\
a_1d_1-a_1t_0^*v_1^*+d_1t_1^*u_1^*\\
a_1d_2-a_1t_0^*v_2^*+d_2t_1^*u_1^*\\
a_2d_1-a_2t_0^*v_1^*+d_1t_1^*u_2^*\\
a_2d_2-a_2t_0^*v_2^*+d_2t_1^*u_2^*
\end{array}
$$
We do a case-by-case analysis of the structure of $I$ and the boundary of $\Omega$.

$ $\\
\noindent \textbf{Case 1: $\omega^* \in \Omega_m$.}
 
This implies $t_0^* \neq 0$ and $t_1^* \neq 0$. Since the indeterminates $b_1,b_2,c_1,c_2$ appear only in the first four polynomials, this suggests the change of variables
$$
\begin{array}{rl}
c_i =& (c_i' -tu_i^*- a_i(t_0^*+t))/(t_1^*-t), \quad i = 1,2\\
b_i =& (b_i' -tv_i^*- d_i(t_1^*-t))/(t_0^*+t), \quad i = 1,2 \\
\end{array}
$$
with new indeterminates $t,a_1,a_2,b_1',b_2',c_1',c_2',d_1,d_2$. In view of Proposition \ref{thm:ChangeOfVar}, the Jacobian determinant of this substitution is a constant. 

$ $\\
\noindent \textbf{Case 1.1: $\omega^* \in \Omega_{m1}$.} 

This implies $u^*\neq 0, v^*=0$ or $u^*= 0, v^*\neq 0$. Without loss of generality, we assume $v^*=0, u_1^*> 0$ and substitute
$$
d_i = (d_i'+a_1t_0^*v_i^*)/(t_1^*u_1^*+a_1), \quad i =1,2. \\
$$ The resulting pullback ideal is $\langle b_1',b_2',c_1',c_2',d_1',d_2'\rangle$. If $\omega^*$ lies in the interior of $\Omega$, we use either Newton polyhedra or Proposition \ref{thm:DisjointVars} to show that the RLCT of this monomial ideal is $(6,1)$. If $\omega^*$ lies on the boundary of $\Omega$, the situation is more complicated. Since we are considering a subset of a neighborhood of $\omega^*$, the corresponding Laplace integral from Proposition \ref{thm:DefRLCT}a is smaller so the threshold is at least $(6,1)$. To compute it exactly, we need blowups to separate the coordinate hyperplanes and the hypersurfaces defining the boundary. 

Because $-u_1^* = u_2^*+u_3^*$, we cannot have $u_2^*=u_3^*=0$. Suppose $u_2^* \neq 0$ and $u_3^* \neq 0$.  We consider a blowup where one of the charts is given by the monomial map $t=s, a_i = sa_i', c_1' = rs, c_2'=rsc_2'', b_i'=rsb_i'', d_i'=rsd_i''$. Here, the pullback pair is $(\langle rs \rangle; r^5s^8)$. Now, we study the inequalities which are \emph{active} at $\omega^*$. For instance, if $b_1^* = 0$, then $\omega^*$ lies on the boundary defined by $0 \leq b_1+b_1^*$. After the various changes of variables, the inequalities are as shown below, where $b_3'' = -b_1''-b_2''$ and similarly for $c_3'',d_3''$ and $a_3'$. Note that the inequality for $a_1^*=0$ is omitted because $a_1^* = 0$ implies $u_1^* = -c_1^* \leq 0$. Similar conditions on the $u_i^*, v_i^*$ hold for the other inequalities.
$$
\begin{array}{rll}
b_i^*=0: & 0 \leq rs(b_i''-d_i''(t_1^*-s)/(t_1^*u_1^*+sa_1'))/(t_0^*+s) & \\
d_i^*=0: & 0 \leq rs d_i''/(t_1^*u_1^*+sa_1') & \\
c_1^*=0: & 0 \leq s(-  u_1^* + a_1'(t_0^*+s)+r)/(t_1^*-s) & \\
c_2^*=0: & 0 \leq s(-  u_2^*+ a_2'(t_0^*+s)+rc_2'')/(t_1^*-s) & u_2^* > 0\\
c_3^*=0: & 0 \leq s(-  u_3^*+ a_3'(t_0^*+s)-r-rc_2'')/(t_1^*-s) & u_3^* > 0\\
a_2^*=0: & 0 \leq sa_2' & u_2^* < 0\\
a_3^*=0: & 0 \leq sa_3' & u_3^* < 0
\end{array}
$$
In applying Lemma \ref{thm:EpsCone}, the choice of coordinates is important. For instance, if $b_2^*=b_3^*=0$, we choose coordinates $b_2''$ and $b_3''$ and set $b_1'' =-b_2''-b_3''$. The same is done for the $d_i''$. The pullback pair is unchanged by these choices. Now, with coordinates $(r,s)$ and $(b_{i_1}'', b_{i_2}'', d_{j_1}'', d_{j_2}'', c_2'', a_2', a_3')$, we apply the lemma with the vector $\xi = (2,2,u_1^*, u_1^*, 1, 1, 1)$, so the threshold is $\RLCT_0(rs;r^5s^8) = (6,1)$.

Now, if only one of $u_2^*, u_3^*$ is zero, suppose $u_2^* = 0, u_3^* \neq 0$ without loss of generality. If $a_2^* =c_2^* \neq 0$, then the arguments of the previous paragraph show that the RLCT is again $(6,1)$. If $a_2^*=c_2^*=0$, we blow up the origin in $\R^7$ and consider the chart where $a_2 = s, c_i' = sc_i'', b_i' = sb_i'', d_i' = sd_i''$. The pullback pair is $(\langle sb_1'', sb_2'', sc_1'', sc_2'', sd_1'', sd_2'' \rangle; s^6)$. The active inequalities for $a_2^*=c_2^*=0$ are
$$
\begin{array}{rll}
c_2^*=0: & 0 \leq s(c_2'' - t_0^*+t)/(t_1^*-t) & \\
a_2^*=0: & 0 \leq s. & \\
\end{array}
$$
Near the origin in $(s, b_1'', b_2'', c_1'', c_2'', d_1'', d_2'') \in \R^7$, these inequalities imply $s=0$ so the new region $\mathcal{M}$ defined by the active inequalities is not full at the origin. Thus, we can ignore the origin in computing the RLCT. All other points on the exceptional divisor of this blowup lie on some other chart of the blowup where the pullback pair is $(s;s^6)$, so the RLCT is at least $(7,1)$. In the chart where $c_2=s, c_1=sc_1'', a_2 = sa_2', b_i' = sb_i'', d_i' = sd_i''$, we have the active inequalities below. Note that $c_3^* \neq 0$ because $u_3^* = -u_1^* < 0$.
$$
\begin{array}{rll}
b_i^*=0: & 0 \leq s(b_i''-d_i''(t_1^*-t)/(t_1^*u_1^*-(sa_2'+a_3))/(t_0^*+t) & \\
d_i^*=0: & 0 \leq s d_i''/(t_1^*u_1^*-(sa_2'+a_3)) &\\
c_1^*=0: & 0 \leq (sc_1'' -  tu_1^* +(sa_2'+a_3)(t_0^*+t))/(t_1^*-t) & \\
c_2^*=0: & 0 \leq s(1 - a_2'(t_0^*+t))/(t_1^*-t) & \\
a_2^*=0: & 0 \leq sa_2' & \\
a_3^*=0: & 0 \leq a_3 &
\end{array}
$$
Again, choosing suitable coordinates in the $b_i''$ and $d_i''$, we
find that the RLCT is  $(7,1)$ by using Lemma
\ref{thm:EpsCone} with $\xi=(2, 2, u_1^*, u_1^*, 1,1,1,-1)$
in coordinates $(b_{i_1}'', b_{i_2}'', d_{j_1}'', d_{j_2}'',  a_2',
a_3, c_1'', t)$.

$ $\\
\noindent \textbf{Case 1.2: $\omega^* \in \Omega_{m2}$.}

This implies $u^*\neq 0, v^*\neq 0$. Without loss of generality, suppose that $u_1^* \neq 0$. If $\omega^* \in \Omega_{m21}$, we further assume that $a_1^*=d_j^*=0, u_1^* \neq 0, v_j^* \neq 0$. Substituting
$$
\begin{array}{rl}
d_i = & (d_i'+a_1t_0^*v_i^*)/(a_1+t_1^*u_1^*), \quad i =1,2 \\
a_2 = & (a_2' + a_1 u_2^*)/u_1^*,
\end{array}
$$
the pullback ideal is $\langle a_2', b_1',b_2',c_1',c_2',d_1',d_2'\rangle$ so the RLCT is at least $(7,1)$. Note that $a_i = (a_2'w_i^* + a_1u_i^*) /u_1^*$ for $i=1,2,3$ where $w_i^* = 0,1,-1$ respectively. If $\omega^*$ is not in $\Omega_{m21}$, we consider the blowup chart $a_2'=s, b_i'=sb_i'', c_i' =sc_i'', d_i'=sd_i''$. The active inequalities are as follows. The symbol $v-$ denotes $v_i^* \leq 0$.
$$
\begin{array}{rll}
b_i^*=0: & 0 \leq [sb_i''-tv_i^*-(sd_i''+a_1t_0^*v_i^*)(t_1^*-t)/(t_1^*u_1^*+a_1)]/(t_0^*+t) & v-\\
c_i^*=0: & 0 \leq [sc_i''- tu_i^* - (sw_i^*+a_1u_i^*)(t_0^*+t)/u_1^*]/(t_1^*-t) & u+\\
a_i^*=0: & 0 \leq  (sw_i^*+a_1u_i^*)/u_1^* & u-\\
d_i^*=0: & 0 \leq (sd_i''+a_1t_0^*v_i^*)/(t_1^*u_1^*+a_1) & v+\\
\end{array}
$$
The crux to understanding the inequalities is this: if $a_i^*=d_j^*=0, u_i^* \neq 0, v_j^* \neq 0$, the coefficient of $a_1$ appears with different signs in the inequalities for $a_i^*=0$ and $d_j^*=0$. This makes it difficult to choose a suitable vector $\xi$ for Lemma \ref{thm:EpsCone}. Similarly, if $b_i^*=c_j^*=0, v_i^* \neq 0, u_j^* \neq 0$, the coefficient of $u_1^*t+t_0^*a_1$ appears with different signs. Fortunately, since $\omega^* \notin \Omega_{m21}$, we do not have such obstructions and it is an easy exercise to find the vector $\xi$. Thus, the RLCT is $(7,1)$.

If $\omega^* \in \Omega_{m21} \setminus \Omega_{m22}$, we blow up $a_1=s, a_2'=sa_2'', b_i'=sb_i'', c_i=sc_i'', d_i=sd_i''$. The active inequalities for $a_1^*=d_j^*=0$ imply that the new region $\mathcal{M}$ is not full at the origin of this chart. Thus, we shift our focus to the other charts of the blowup where the pullback pair is $(s;s^7)$, so the RLCT is at least $(8,1)$. In the chart where $a_2'=s, a_1=sa_1', b_i'=sb_i'', c_i=sc_i'', d_i=sd_i''$, we do not have obstructions coming from any $b_i^*=c_j^*=0, v_i^* \neq 0, u_j^* \neq 0$ so it is again easy to find the vector $\xi$ for Lemma \ref{thm:EpsCone}. The threshold is exactly $(8,1)$.

If $\omega^* \in \Omega_{m22}$, consider the following two charts out of the nine charts in the blowup of the origin in $\R^9$.
$$
\begin{array}{rl}
\mbox{Chart 1:} & a_1=s, t=st', a_2'=sa_2'', b_i'=sb_i'', c_i=sc_i'', d_i=sd_i'' \\
\mbox{Chart 2:} & t=s, a_1=sa_1', a_2'=sa_2'', b_i'=sb_i'', c_i=sc_i'', d_i=sd_i'' \\
\end{array}
$$
The inequalities for $a_i^*=d_j^*=0, u_i^* \neq 0, v_j^* \neq 0$ and $b_i^*=c_j^*=0, v_i^* \neq 0, u_j^* \neq 0$ imply that the new region $\mathcal{M}$ is not full at points outside of the other seven charts, so we may ignore these two charts in computing the RLCT. Indeed, for Chart 1, the active inequalities
$$
\begin{array}{rll}
a_i^*=0: & 0 \leq  s(a_2''w_i^*+u_i^*)/u_1^* & u-\\
d_i^*=0: & 0 \leq s(d_i''+t_0^*v_i^*)/(t_1^*u_1^*+s) & v+\\
\end{array}
$$
tell us that $a_2''$ or $d_2''$ must be non-zero for $\mathcal{M}$ to be full. In Chart 2, suppose $\mathcal{M}$ is full at some point $x$ where $a_2''=b_1''=b_2''=c_1''=c_2''=d_1''=d_2''=0$. Then, 
$$
\begin{array}{rll}
a_i^*=0: & 0 \leq  s(a_2''w_i^*+a_1'u_i^*)/u_1^* & u-\\
d_i^*=0: & 0 \leq s(d_i''+a_1't_0^*v_i^*)/(t_1^*u_1^*+sa_1') & v+\\
\end{array}
$$
imply that $a_1'=0$ at $x$. However, if this is the case, the inequalities
$$
\begin{array}{rll}
b_i^*=0: & 0 \leq s[b_i''-v_i^*-(d_i''+a_1't_0^*v_i^*)(t_1^*-s)/(t_1^*u_1^*+sa_1')]/(t_0^*+s) & v-\\
c_i^*=0: & 0 \leq s[c_i''- u_i^* - (a_2''w_i^*+a_1'u_i^*)(t_0^*+s)/u_1^*]/(t_1^*-s) & u+\\
\end{array}
$$
forces $b_i''$ or $c_i''$ to be non-zero for some $i$, a contradiction. Thus, we shift our~focus to the other seven charts where the pullback pair is $(s;s^8)$ and the RLCT is at least $(9,1)$. In the chart for $a_2'=s, a_1 = sa_1', t=st', b_i'=sb_i'', c_i'=sc_i'', d_i'=sd_i''$, note that we cannot have both $a_2^*=0$ and $a_3^*=0$ because we assumed $a_1^*=0$. It is now easy to find the vector $\xi$ for Lemma \ref{thm:EpsCone}, so the threshold is $(9,1)$.

$ $\\
\noindent \textbf{Case 1.3: $\omega^* \in \Omega_{m0}$.} 

This implies $u_i^* = v_i^* =0$ for all $i$. The pullback ideal can be written as $$\langle b_1', b_2', c_1', c_2' \rangle + \langle a_1, a_2 \rangle \langle d_1, d_2\rangle$$ 
whose RLCT over an interior point of $\Omega$ is $(6,2)$ by Proposition \ref{thm:DisjointVars}. This occurs in $\Omega_{m000}$ where none of the inequalities are active. Now, suppose the only active inequalities come from $a_1^* = c_1^*=0$. We blow up the origin in $\{(a_1, c_1')\in \R^2\}$. In the chart given by $a_1=a_1', c_1'=a_1'c_1''$, the new region $\mathcal{M}$ is not full at the origin, so we only need to study the chart where $c_1' = c_1'', a_1 = c_1''a_1'$. The pullback pair becomes $(\langle c_1'' \rangle + \langle b_1', b_2', c_2' \rangle + \langle a_2 \rangle \langle d_1, d_2 \rangle; c_1'')$, and a simple application of Lemma \ref{thm:EpsCone} and Proposition \ref{thm:DisjointVars} shows that the threshold is $(6,1)$.

In this fashion, we study the different scenarios and summarize the pullback pairs and thresholds in the table below.
$$
\begin{array}{llll}
\hline
\mbox{Inequalities}  & \mbox{Pullback pair} & & \mbox{RLCT} \\
\hline
\vspace{-0.1in} &&& \\
- & (\langle b_1', b_2', c_1',c_2' \rangle + \langle a_1,a_2 \rangle \langle d_1, d_2 \rangle; &1) & (6,2) \\
a_1^*=0 & (\langle b_1', b_2', c_1'',c_2' \rangle + \langle a_2 \rangle \langle d_1, d_2 \rangle; &c_1'') & (6,1) \\
a_1^*=0, b_1^*=0 & (\langle b_1'', b_2', c_1'', c_2' \rangle + \langle a_2 \rangle \langle d_2 \rangle; &b_1''c_1'') & (7,2)\\
a_1^*=1 & (\langle b_1', b_2',c_1'',c_2'' \rangle; &c_1''c_2'') & (6,1) \\
a_1^*=1, b_1^*=0 & (\langle b_1'', b_2', c_1'', c_2'' \rangle; &b_1''c_1''c_2'') & (7,1)\\
a_1^*=1, b_1^*=1 & (\langle b_1'', b_2'', c_1'', c_2'' \rangle; &b_1''b_2''c_1''c_2'') & (8,1)\\
\vspace{-0.1in} &&& \\
\hline
\end{array}
$$
For example, the case $a_3^* = c_3^*=1$ corresponds to $a_1^*=a_2^* = c_1^* = c_2^* = 0$. Here, we blow up the origins in  $\{(a_1, c_1')\in \R^2\}$ and  $\{(a_2, c_2')\in \R^2\}$. As before, we~can ignore the other charts and just consider the one where $a_1=c_1''a_1', c_1' = c_1'', a_2=c_2''a_2', c_2' = c_2''$. The pullback pair is $(\langle c_1'' \rangle + \langle c_2'' \rangle + \langle b_1', b_2'\rangle, c_1''c_2'')$. If $b_i^* \neq 0$ for all $i$, the RLCT is $(6,1)$ by Lemma \ref{thm:EpsCone} and Proposition \ref{thm:DisjointVars}.

$ $\\
\noindent \textbf{Case 2: $\omega^* \in \Omega_{u}$.} 

Without loss of generality, assume $t^*=0$ and substitute
$$
\begin{array}{rll}
c_i &= (c_i'- t(a_i+u_i^*))/(1-t) &i = 1,2\\
d_i &= (d_i'- t(b_i+v_i^*))/(1-t) &i = 1,2.\\
\end{array}
$$
The pullback ideal is the sum of $\langle c_1', c_2', d_1', d_2' \rangle$ and
$$
\langle t \rangle \langle a_1+u_1^*, a_2+u_2^* \rangle \langle b_1+v_1^*,b_2+v_2^*\rangle .
$$
Since $c_3' = -c_1'-c_2'$ and similarly for the $d_i', a_i, b_i, u_i^*$ and $v_i^*$, it is useful to write this ideal more symmetrically as the sum of $\langle c_1', c_2', c_3' \rangle$, $\langle d_1', d_2', d_3' \rangle$ and
$$
\langle t \rangle \langle a_1+u_1^*, a_2+u_2^*, a_3+u_3^* \rangle \langle b_1+v_1^*,b_2+v_2^*, b_3+v_3^*\rangle .
$$
Meanwhile, the inequalities are
$$
\begin{array}{rll}
a_i^* = 0: & 0 \leq a_i & \\
c_i^*=0: & 0 \leq  (c_i'- t(a_i+u_i^*))/(1-t) & u_i^* \geq 0 \\
b_j^* = 0: & 0 \leq b_j & \\
d_j^*=0: & 0 \leq (d_j'- t(b_j+v_j^*))/(1-t) & v_j^* \geq 0. \\
\end{array}
$$
We now relabel the indices of the $a_i$ and $c_i'$, without changing the $b_j$ and $d_j'$, so that the active inequalities are among those from $a_1^* =0, a_2^*=0, c_{i_1}^* = 0, c_{i_2}^* = 0$. The $b_j$ and $d_j'$ are thereafter also relabeled so that the inequalities come from $b_1^* =0, b_2^*=0, d_{j_1}^* = 0, d_{j_2}^* =0$. We claim that the new region $\mathcal{M}$ contains, for small $\varepsilon$, the orthant neighborhood
$$\{(a_1, a_2, b_1, b_2, c_{i_1}, c_{i_2}, d_{j_1}, d_{j_2}, -t) \in [0, \varepsilon]^9\}.$$ Indeed, the only problematic inequalities are
$$
\begin{array}{rll}
c_3^*=0: & 0 \leq  (c_3'- t(-a_1-a_2+u_i^*))/(1-t) & u_3^* = 0 \\
d_3^*=0: & 0 \leq (d_3'- t(-b_1-b_2+v_j^*))/(1-t) & v_3^* = 0. \\
\end{array}
$$
However, these inequalities cannot occur because for instance,
$u_3^*=0$ and $c_3^*=0$ implies $a_3^*=0$, a contradiction since the
$a_i$ were relabeled to avoid this. Finally, the threshold of $\langle
t \rangle$ is $(1,1)$ while that of $\langle a_1+u_1^*, a_2+u_2^*
\rangle$ and $\langle b_1+v_1^*,b_2+v_2^*\rangle$ are at least $(2,1)$
each. By Proposition~\ref{thm:DisjointVars}, the RLCT of their product
is $(1,1)$ and that of the pullback ideal we were originally
interested in is $(5,1)$.
\end{proof}

\bigskip \bigskip

\noindent {\bf Acknowledgements.} The author wishes to thank Christine Berkesch, Mathias Drton, Anton Leykin, Bernd Sturmfels, Zach Teitler, Sumio Watanabe and Piotr Zwiernik, as well as the anonymous reviewers for their many useful suggestions, discussions and corrections.


\bigskip






\begin{thebibliography}{2}

\bibitem{AW}
M.~Aoyagi and S.~Watanabe: Stochastic complexities of reduced rank regression in Bayesian estimation, {\em Neural Networks} {\bf 18} (2005) 924--933. 

\bibitem{AGV} V.~I.~Arnol'd, S.~M.~Guse\u\i n-Zade and A.~N.~Varchenko: {\em Singularities of Differentiable Maps}, Vol. II, Birkh\"auser, Boston, 1985.

\bibitem{BBO} J.~Bertrand, P.~Bertrand and J.~Ovarlez: The Mellin Transform, in {\em The Transforms and Applications Handbook: Second Edition}, Chapter 12, Ed. A. D. Poularikas, CRC Press, Boca Raton, 2010.

\bibitem{BM}  E.~Bierstone and P.~D.~Milman: Resolution of singularities, {\em Several complex variables}, MSRI Publications {\bf 37} (1999) 43--78. 

\bibitem{BiviaAusina} C.~Bivi\`a-Ausina: Nondegenerate ideals in formal power series rings, {\em Rocky Mountain 
J. Math.} {\bf 34} (2004) 495--511. 

\bibitem{BL} M.~Blickle and R.~Lazarsfeld: An informal introduction to multiplier ideals, {\em Trends in Commutative Algebra}, MSRI Publications {\bf 51} (2004) 87--114.

\bibitem{BEV} A.~Bravo, S.~Encinas and O.~Villamayor: A simplified proof of desingularisation and
applications, {\em Rev. Math. Iberoamericana} {\bf 21} (2005) 349--458.

\bibitem{DP} M. Drton and M. Plummer: A Bayesian Information Criterion for Singular Models, {\em J. R. Statist. Soc.} B {\bf 79} (2017) 1--38.

\bibitem{DSS} M.~Drton, B.~Sturmfels and S.~Sullivant: {\em Lectures on Algebraic Statistics}, Oberwolfach Seminars {\bf 39}, Birkh\"auser, Basel, 2009. 

\bibitem{E} D.~Eisenbud: {\em Commutative Algebra with a view towards Algebraic Geometry}, Graduate Texts in Mathematics {\bf 150}, Springer-Verlag, New York, 1995.

\bibitem{EGG}  M.~Evans, Z.~Gilula and I.~Guttman:  Latent class analysis of two-way contingency
tables by Bayesian methods, {\em Biometrika} {\bf 76} (1989) 557--563.

\bibitem{F} W.~Fulton: {\em Introduction to Toric Varieties}, Annals of Mathematics Studies {\bf 131}, Princeton University Press, Princeton, 1993.

\bibitem{GR} D.~Geiger and D.~Rusakov: Asymptotic model selection for naive Bayesian networks, {\em J. Mach. Learn. Res.} {\bf 6} (2005)~1--35.

\bibitem{G2} M.~Greenblatt: An elementary coordinate-dependent local resolution of singularities and applications, {\em J. Funct. Anal.} {\bf 255} (2008) 1957--1994.

\bibitem{G} M.~Greenblatt: Resolution of singularities, asymptotic expansions of integrals, and applications, {\em J. Analyse Math.} {\bf 111} (2010) 221--245. 

\bibitem{Hi} H.~Hironaka: Resolution of singularities of an algebraic variety over a field of characteristic zero I, II, {\em Ann. of Math.} (2) {\bf 79} (1964) 109--203.

\bibitem{H} J.~A.~Howald: Multiplier ideals of monomial ideals, {\em Trans. Amer. Math. Soc} {\bf 353} (2001) 2665--2671.

\bibitem{K} J.~Koll{\'a}r: Singularities of pairs, in {\em Algebraic geometry---Santa Cruz 1995}, 221--287, Proc. Symp. Pure Math. {\bf 62}, Amer. Math. Soc., Providence, 1997.

\bibitem{K2} J.~Koll{\'a}r: {\em Lectures on Resolution of Singularities} (AM-166), Princeton University Press, 2009.

\bibitem{L}  R.~Lazarsfeld: {\em Positivity in Algebraic Geometry I, II}, A Series of Modern Surveys in Mathematics {\bf 48--49}, Springer-Verlag, Berlin, 2004.

\bibitem{LSX} S.~Lin, B.~Sturmfels and Z.~Xu: Marginal likelihood integrals for
mixtures of independence models, {\em J. Mach. Learn. Res.} {\bf 10} (2009)~1611--1631. 

\bibitem{Saia} M.~J.~Saia: The integral closure of ideals and the Newton filtration, {\em J. Algebraic 
Geom.} {\bf 5} (1996) 1--11. 

\bibitem{S} M.~Saito: On real log canonical thresholds, {\tt arXiv:math.AG/0707.2308}.

\bibitem{Stu} B.~Sturmfels: {\em Gr\"{o}bner Bases and Convex Polytopes}, University Lecture Series {\bf 8}, Amer. Math. Soc., Providence, 1996.

\bibitem{V} A.~N.~Varchenko: Newton polyhedra and estimation of oscillating integrals, {\em Funct. Anal. Appl.} {\bf 10} (1977) 175--196.

\bibitem{Vasilev} V.~A.~Vasil'ev: Asymptotic behavior of exponential integrals in the complex domain, {\em Funktsional. Anal. i Prilozhen.} {\bf 13}:4 (1979) 1--12. 

\bibitem{W0} S.~Watanabe: Algebraic analysis for nonidentifiable learning machines, {\em Neural Computation} {\bf 13} (2001) 899--933.

\bibitem{W} S.~Watanabe: {\em Algebraic Geometry and Statistical Learning Theory}, Cambridge Monographs on Applied and Computational Mathematics {\bf 25}, Cambridge University Press, Cambridge, 2009.

\bibitem{YW1} K.~Yamazaki and S.~Watanabe: Singularities in mixture models and upper bounds of stochastic complexity,  {\em  International Journal of Neural Networks} {\bf 16} (2003) 1029--1038.

\bibitem{YW} K.~Yamazaki and S.~Watanabe: Newton diagram and stochastic complexity in mixture of binomial distributions, {\em Algorithmic Learning Theory}, 350--364, Lecture Notes in Comput. Sci. {\bf 3244}, Springer, Berlin, 2004. 

\end{thebibliography}
\end{document}